\documentclass[11pt,reqno]{amsart}

\usepackage[dvipsnames]{xcolor}
\usepackage{comment}

\newtheorem{theorem}{Theorem}[section]
\newtheorem{proposition}[theorem]{Proposition}
\newtheorem{lemma}[theorem]{Lemma}

\theoremstyle{definition}

\newtheorem{remark}[theorem]{Remark}

\numberwithin{equation}{section}

\definecolor{ReferenceGreen}{HTML}{1F6B4F}
\definecolor{CitationRed}{HTML}{A23B3B}
\definecolor{URLBlue}{HTML}{315E85}

\usepackage[
  colorlinks=true,
  linkcolor=ReferenceGreen,
  citecolor=CitationRed,
  urlcolor=URLBlue,
  pdfborder={0 0 0}
]{hyperref}

\usepackage[
  nameinlink,
  capitalize,
  noabbrev
]{cleveref}

\usepackage[
  left=27mm,
  right=27mm,
  top=30mm,
  bottom=32mm
]{geometry}

\usepackage{amssymb,mathtools,mathrsfs}
\usepackage{enumitem}
\usepackage{booktabs,tabularx,array}
\usepackage{microtype}

\IfFileExists{newtxtext.sty}{%
  \usepackage{newtxtext,newtxmath}
}{%
  \usepackage{mathptmx}
}

\usepackage{parskip}
\newcommand{\bea}{\begin{eqnarray}}
\newcommand{\eea}{\end{eqnarray}}
\newcommand{\<}{\langle}
\renewcommand{\>}{\rangle}

\newcommand\eg{{\text{\eg~}}}

\def\eps{{\varepsilon}}

\def\spec{{\rm spec}}

\def\cF{{\mathcal F}}

\def\R{{\mathbb R}}

\def\de{{\rm d}}

\def\<{\langle}
\def\>{\rangle}

\def\cN{{\cal N}}

\def\cD{{\cal D}}

\def\b0{{\boldsymbol{0}}}

\def\spec{{\mathrm {spec}}}

\def\de{{\mathrm {d}}}

\def\cD{{{\mathcal D}}}

\def\cI{{\mathcal I}}

\renewcommand{\b}{\mathbf{b}}

\def\lt{\left}
\def\rt{\right}

\def\eps{\varepsilon}

\def\bbC{{\mathbb{C}}}

\def\bbR{{\mathbb{R}}}

\def\cF{{\mathcal{F}}}

\def\cN{{\mathcal{N}}}

\def\cQ{{\mathcal{Q}}}

\def\sC{{\mathscr{C}}}

\def\one{{\mathbbm{1}}}

\newcommand{\norm}[1]{{\lt\|#1\rt\|}}

\def\F{{\boldsymbol{F}}}

\newcommand{\ip}[2]{\left\langle #1,#2\right\rangle}

\renewcommand{\F}{\mathcal F}

\newcommand{\N}{\mathcal N}
\newcommand{\Pf}{P_{\mathrm f}}

\renewcommand{\one}{\mathbf 1}
\renewcommand{\spec}{\operatorname{spec}}
\newcommand{\essspec}{\operatorname{spec}_{\mathrm{ess}}}
\newcommand{\Ran}{\operatorname{Ran}}
\newcommand{\Ker}{\operatorname{Ker}}
\newcommand{\dd}{\,\mathrm d}
\renewcommand{\cD}{\mathcal D}
\renewcommand{\cQ}{\mathcal Q}
\renewcommand{\eps}{\varepsilon}

\newcommand{\Ialpha}{\mathcal I_\alpha}

\title{The Fr\"ohlich Polaron at large total momentum in $d=3$}
\author{Tobias Schmidt}
\address{Center of Mathematical Sciences and Applications, Harvard University, USA.}
\email{tobias.schmidt@cmsa.fas.harvard.edu}
\date{\today}

\begin{document}

\begin{abstract}
We consider the three-dimensional Fr\"ohlich polaron at fixed coupling strength $\alpha>0$ and total momentum $P \in \bbR^3$. We show that the bottom of the spectrum reaches the essential spectrum at a finite value of $|P|$, implying that the fiber Hamiltonian has no ground state at sufficiently large total momentum. This extends the known weak-coupling result to arbitrary coupling. Our approach describes states near the spectral threshold by a fully interacting polaron with small residual momentum and one additional phonon carrying almost all of the total momentum.
\end{abstract}

\maketitle

\section{Introduction and main results}
The Fr\"ohlich polaron models an electron interacting with the optical phonons of a polar crystal. The electron and the surrounding phonon cloud form a quasi-particle, whose energy depends on their conserved total momentum. In this paper, we study whether this quantum system has a ground state when the prescribed total momentum is large. We prove that, in three dimensions and at every fixed coupling strength, ground states do not exist for momenta sufficiently large. This problem goes back at least to works by Feynman \cite{Fe72} and later by Spohn \cite{Spohn1988}, where it was shown that a drastically different behavior emerges in $d=1,2$.

Let $\alpha>0$ denote the coupling strength between the phonons and the electron. At fixed total momentum $P\in\R^3$, the Fr\"ohlich Hamiltonian is given by
\begin{equation}\label{eq:intro-H}
 H(P,\alpha)=\frac12(P-\Pf)^2+\N+\Phi(g_\alpha),
 \qquad
 g_\alpha(k)=\frac{\sqrt\alpha}{\sqrt2\,\pi|k|}.
\end{equation}
In the above, $P_{\mathrm f}$ is the momentum of the phonon field, $\cN$ the phonon number operator, and $\Phi(g_\alpha)=a(g_\alpha)+a^*(g_\alpha)$ is the Segal field; see also Section \ref{subsec:fock-space-preliminaries}. In our convention the electron mass and the energy of a single phonon are equal to $1$. Since $g_\alpha$ is not square integrable at infinity, \eqref{eq:intro-H} is understood in the sense of quadratic forms. We recall the precise construction in Appendix~\ref{app:basic_froehlich} and write
\[
 E(P,\alpha)=\inf\spec H(P,\alpha).
\]
A ground state at momentum $P$ is defined to be an eigenvector with norm $1$ of $H(P,\alpha)$, with eigenvalue $E(P,\alpha)$. Suppressing $\alpha$ from the notation, we denote the ground state by $\psi_P$. Whenever it exists, it is uniquely determined up to a phase (i.e. the eigenspace of the lowest eigenvalue is $1$-dimensional) \cite{DybalskiSpohn}. The constant phonon dispersion-relation in the Fröhlich model has a simple consequence: a phonon can carry arbitrarily large momentum at an energy cost of $1$. The system can therefore place its momentum in an additional phonon while the remaining polaron stays close to rest. Separating this phonon from the electron leads to the energy $E(0,\alpha)+1$. The corresponding spectral statement is
\[
 \spec_{\mathrm{ess}}H(P,\alpha)=[E(0,\alpha)+1,\infty),
\]
valid for every $P\in\R^3$; see \cite[Lemma~1.1]{LampartMitrouskasMysliwy}. In particular, $E(P,\alpha)\leq E(0,\alpha)+1$. For small $|P|$, the energy is known to lie strictly below this threshold and is a simple, isolated eigenvalue \cite{DybalskiSpohn}. The question is whether the interaction can keep the energy below the threshold at arbitrarily large total momentum. More precisely, identifying the radial function $E(\,\cdot\,,\alpha)$ with a function of $p=|P|$, we consider
\[
 \cI_\alpha
 =\{p\geq0:E(p,\alpha)<E(0,\alpha)+1\}.
\]
Monotonicity of the energy implies that $\cI_\alpha$ is an interval containing a neighborhood of zero \cite{Polzer2023}. Its boundedness was established at sufficiently small coupling (meaning $\alpha \ll 1$) in \cite{Dahlbaek}; the question at arbitrary coupling is also discussed in \cite{Polzer2026}. Our main result is the following.
\begin{theorem}\label{thm:main}
For every $\alpha>0$, there exists $P_\alpha<\infty$ such that, for every $P\in\R^3$ with $|P|\ge P_\alpha$,
\begin{equation}\label{eq:main_goal}
 E(P,\alpha)=E(0,\alpha)+1.
\end{equation}
Moreover,
\begin{equation}\nonumber
 \Ker\bigl(H(P,\alpha)-E(0,\alpha)-1\bigr)=\{0\}.
\end{equation}
In particular, $H(P,\alpha)$ has no ground state at sufficiently large total momentum, and the interval $\Ialpha$ is bounded.
\end{theorem}
The fact that no ground state can exist at the spectral edge for momenta outside the closure of $\cI_\alpha$ is already known (see e.g. \cite{Polzer2023}). We included the statement again because we provide an alternative proof that we find interesting in its own right, and because it does not require much effort. The quantifiers in Theorem~\ref{thm:main} should be noted: $\alpha$ is fixed first, and the required threshold $P_\alpha$ may depend strongly on $\alpha$. We unfortunately do not obtain a useful estimate on $P_\alpha$. Quantitative bounds would be interesting for obtaining a better understanding of this transition; see also \cite{HinrichsLampart2024} for a  discussion related to similar concepts in the Bose polaron and the so-called Cherenkov transition. We also note that our proof does not use the independent gap-closing result
\[
 E(P,\alpha)\longrightarrow E(0,\alpha)+1
 \qquad(|P|\to\infty),
\]
which can be deduced by combining the cutoff gap-closing result from \cite{Moller} together with uniform cutoff convergences of the $E_\Lambda (P) \to E(P)$ obtained in \cite{LampartMitrouskasMysliwy}. Instead, we use the Rayleigh-Ritz principle and spectral estimates to show that all states have an energy bounded below by $E(0,\alpha)+1$ for sufficiently large $|P|$. This, together with the well-known inequality $E(P,\alpha) \le E(0,\alpha)+1$ then shows that there exists a finite point of saturation, meaning $E(P,\alpha) = E(0,\alpha)+1$ is achieved eventually.

\begin{remark}
    Our proof is essentially dimension-independent, and only requires integrability of $|k|^{-2}$ around $0$. We are confident that an extension of our approach to $d >3$ is possible in case one considers the coupling function (up to a constant chosen so that the Fourier transform of its square evaluates to $\alpha/|x|$)
    $$g(k) = \sqrt{\alpha }|k|^{-(d-1)/2};$$
    see also \cite{PeetersWuDevreese1986,Moller}, where this generalization of the three-dimensional Fr\"ohlich model was introduced.
    We are also convinced that the critical threshold $P_\alpha$ is non-increasing in the dimension, which can be seen using path integral methods. Since the inputs \ref{I:hvz}-\ref{I:gross} are not all available in general dimensions, we postpone an extension to higher dimensions to future work. 
\end{remark}

\subsection{Related work}
The large-momentum problem has a long history, see in particular the classical works \cite{Fe72,Spohn1988,Moller}. In \cite{Moller}, the author established spectral and large-momentum results for a class of translation-invariant particle-field models with square-integrable coupling functions. In particular, his results imply convergence of the energy to the continuum threshold for suitable ultraviolet regularizations of the Fr\"ohlich model
\cite[Theorem~2.4]{Moller}. Compared to that, the work \cite{Spohn1988} studied the Fr\"ohlich model in $d=1,2$. Spohn established that a spectral gap always persists, meaning that a ground state exists at arbitrarily large momentum. He noted that different behavior can be expected at $d=3$ at large values of $|P|$, which is exactly what we verify.  For the physical Fr\"ohlich interaction in three dimensions, the closest predecessor to
Theorem~\ref{thm:main} is the weak-coupling analysis in \cite{Dahlbaek}. The author proved that there are no discrete eigenvalues below the essential spectrum at sufficiently large total momentum when the coupling is sufficiently small; see \cite[Theorem~1.3(2)(d)]{Dahlbaek}. His approach is perturbative in nature and reduces the spectral problem to generalized Friedrichs models by means of Feshbach operators and the Birman-Schwinger principle. Theorem \ref{thm:main} removes the small-coupling assumption for the Fr\"ohlich model. The corresponding low-energy states in our argument contain a fully interacting polaron with small residual momentum instead of the Fock vacuum and one additional phonon carrying almost all of the total momentum.

Related one-boson constructions to the one we present in this work appear in spectral and scattering theory. For example, the work \cite{AngelescuMinlosZagrebnov2005} studies particle-field Hamiltonians in dimensions $d\geq3$ at weak coupling, under smoothness and rapid-decay assumptions on the interaction. Below a prescribed energy strictly below the two-boson threshold, they constructed invariant subspaces describing a polaron and one additional boson, on which the Hamiltonian is
unitarily equivalent to a multiplication operator. 
In \cite{MollerRasmussen2013}, the continuous spectrum below the two-boson threshold for a class of translation-invariant massive particle-field models is studied. \cite{DybalskiMoller2015} subsequently proved asymptotic completeness in this energy region, without restrictions on the coupling strength or total momentum. Their class includes suitable ultraviolet regularizations of the Fr\"ohlich model. The scattering construction combines a dressed particle and an additional boson through an identification map. Restricted to the ground-state branch at fixed total momentum and to small residual momenta, this map has the form of $J_P$ used below. The associated channel energy is the sum of the polaron energy and the boson energy, namely $E(q,\alpha)+1$ in the Fr\"ohlich case.
Our proof uses a similar construction for a stationary estimate at large total momentum with the full Fr\"ohlich interaction. For every fixed coupling, we control the orthogonal projection onto the low-energy states and the coupling to the orthogonal complement relative to the energy $E(q)-E(0)$, while also establishing a uniform positive gap on that complement. These estimates give the exact lower bound $E(0)+1$ at sufficiently large momentum. 

We rely on the substantial amount of work on the energy-momentum relation of the polaron, as several structural properties are essential for our proof.  \cite{LampartMitrouskasMysliwy} proved that the energy has a strict global minimum at momentum zero. In the work \cite{DybalskiSpohn} the analytic dependence of the isolated ground state on momentum and the relation between the effective mass and the diffusion constant of the polaron path measure are discussed. Polzer \cite{Polzer2023} proved monotonicity of the energy, concavity as a function of $|P|^2$, and finiteness of the effective mass at every coupling. He also proved that there is no ground state when $|P|$ lies outside the closure of $\mathcal I_\alpha$. Thus, once the boundedness of this interval is established, absence of ground states at larger momenta already follows from his result. Our argument proves boundedness and directly excludes a threshold eigenvector for all sufficiently large $|P|$. More recently, \cite{DesioSeiringer} also developed a Dyson expansion for relatively form-bounded perturbations and extended the concavity result to a larger class of polaron-type models. In \cite{Polzer2026}, the author additionally obtained lower bounds on the energy-momentum relation in terms of the vacuum overlap and the spectral gap. A probabilistic criterion equivalent to the boundedness of $\cI_\alpha$ was derived, as well. These results provide further information about the energy-momentum relation but do not establish finite-momentum saturation for the Fr\"ohlich model at arbitrary coupling.

Another direction of research concerns the strong-coupling limit. Brooks and Seiringer \cite{BrSe24} proved an asymptotically sharp lower bound on the energy-momentum relation. Together with the corresponding upper bound, this identifies its approximately parabolic shape below the continuum threshold and the Landau-Pekar effective-mass scale. 
Determining the dependence of the critical momentum on $\alpha$ would connect these two questions more closely. We have unfortunately nothing to say on the problem, however.

Finally, one can ask related questions about a critical momentum also for an impurity in a Bose-Einstein condensate. The work
\cite{HinrichsLampart2024} proved a lower bound on the momentum up to which a ground state exists in the Bogoliubov-Fr\"ohlich model. This provides related motivation for studying the loss of a stable quasiparticle at large momentum. The dispersion relation in that model is different, however, and its massless excitations lead to a different threshold problem.

\subsection{Proof strategy}
In this section we want to briefly detail our approach to verify Theorem \ref{thm:main}. Here, we draw inspiration from two previous works on the Fr\"ohlich model. First, we were motivated by the work \cite{BrSe24}, in which the authors verify the asymptotically correct upper bound scaling of the effective mass for the Fr\"ohlich model by first restricting all the trial states they need to consider to spatial mixtures of Pekar minimizers. Second, the work \cite{Dahlbaek} used analytic perturbation theory to reduce the Fr\"ohlich model in the small coupling limit $\alpha \to 0$ to a generalized Friedrichs model by arguing that all low-energy states look like a single hard phonon\footnote{By a hard phonon we simply mean a phonon with (comparably) large momentum. We use the term 'hard' in analogy to the physics literature}. The main difficulty is that at coupling strength of $O(1)$ such states are not expected to capture the correct low energy states, because the electron should carry a polarization cloud around with itself. It can, however, be expected that almost all of the momentum should be carried by a single hard phonon, with the remainder of the system sitting in a low-momentum polaron. One is therefore led to consider states of the form
\begin{equation}
    J_Pf =\label{equ:low_energy_states}
    \int \de q f(q) a^*(P-q) \psi_q,
\end{equation}
where $f: \bbR^3 \to \bbC$ should be thought of as a mixture function narrowly concentrated at $0$. Formally taking $f= \delta_0$ recovers the single hard phonon carrying all momentum of the system. Let now $\sC_P$ be the space of all of those vectors with $f \in L^2 (B_r)$. Heuristically speaking 
$$H(P,\alpha) J_P f = J_P fh + O(1/|P|),$$
where $h(q) = E(q,\alpha) +1$ by a direct calculation. Ignoring the error term for a second, this means that $H(P,\alpha)$ leaves $\sC_P$ invariant.
Next, assume that the projection onto $\sC_P$ is simply given by $\Pi_P = J_P J_P ^*$. If one could now also show that $H(P,\alpha)$ is coercive on $\sC_P ^\perp$ with constant at least $E(0,\alpha)+1+\epsilon$, one obtains for any normalized Fock vector $\phi$
\begin{equation}
    \label{equ:heuristic_energy_calculation}
    \begin{split}
        \nonumber
        \langle \phi, H(P,\alpha) \phi \rangle &= \langle \Pi_P \phi, H(P,\alpha) \Pi_P \phi    \rangle + \langle (1-\Pi_P) \phi , H(P,\alpha) (1-\Pi_P) \phi \rangle \\
        &= (1-\delta)\int \de q |f_{\phi}(q)|^2 h(q) + \langle (1-\Pi_P) \phi , H(P,\alpha) (1-\Pi_P) \phi \rangle \\
        &\ge (1-\delta)\int \de q |f_{\phi}(q)|^2 h(q) + \delta (E(0,\alpha)+1+\epsilon), 
    \end{split}
\end{equation}
where $\delta = \Vert (1-\Pi_P)\phi \Vert^2 \in [0,1]$ and $f_\phi$ is the normalized mixing function associated to $\phi$. Because $E(q,\alpha) \ge E(0,\alpha)$ as described earlier, we find
$$\langle \phi,H(P,\alpha) \phi \rangle \ge (1-\delta)(E(0,\alpha)+1) + \delta ( E(0,\alpha) + 1 + \epsilon) \ge E(0,\alpha)+1,$$
verifying that the spectral gap closes at finite $|P|$. The fact that no ground state can exist at $E(0) +1$ can also be inferred from \eqref{equ:low_energy_states}: because $E(q,\alpha) > E(0,\alpha)$ on $B_r (0)\backslash \{0\}$, only $f= \delta_0$ can attain the optimal spectral value.

The trouble is that $J_PJ_P^*\neq1$ on $\sC_P$, meaning that computing the energy of the (presumably) low-energy states is not as easy as the simplified calculation above suggests. We solve this problem in Section \ref{sec:projection_construction} and \ref{sec:channel_energy_estimates}, where it is shown that $J_P J_P ^*$ is almost the projection and energies can be calculated fairly explicitly. We obtain multiplicative errors that become small in the large total momentum regime, which is sufficient for our purposes. Second, a major problem is that $H(P,\alpha)$ couples $\sC_P$ to its orthogonal complement, so even small errors may lower the energy below the threshold $E(0,\alpha)+1$. We control these errors relative to the energies above the threshold, with coefficients tending to zero as $|P|\to\infty$, allowing us to absorb them into the diagonal terms. The obtained bounds rely in an essential manner on the local integrability of $|q|^{-2}$, which holds precisely for $d\ge3$; for further details we refer to Section \ref{sec:channel_energy_estimates}.
Finally, it is also a non-trivial undertaking to lower-bound $H(P,\alpha)$ on $\sC_P ^\perp$ (see Section \ref{sec:orthogonal_to_channel}). Our strategy will be to analyze subspaces of $\sC_P ^\perp$ and verify the increased energy on each of them separately. To this end, we first introduce a large ultraviolet cutoff and separate phonons with momenta below the threshold and above it. The latter do not interact with the electron in the correspondingly chosen cutoff model. $\sC^\perp _{\Lambda,P}$, meaning $\sC^\perp _P$ taken with cutoff ground states $\psi_{\Lambda,q}$, then splits into four sectors, each with a different reason for having an energy gap. First, if no phonon lies above the cutoff, a large total momentum requires either a large electron kinetic energy or many phonons below the cutoff. Second, if at least two phonons lie above the cutoff, their number already costs an additional unit of energy relative to the one-phonon threshold, disqualifying them as proper low-energy states. It remains to consider states with exactly one phonon above the cutoff. When the remaining momentum not carried by the outer phonon lies outside a small neighborhood of zero, the energy gap follows from the strict minimum of the polaron energy at zero. When the polaron momentum lies inside this neighborhood, orthogonality to $\sC_{\Lambda,P}$ forces the remaining state to be orthogonal to the local polaron ground state, and the excitation gap at small momentum supplies the required lower bound. Finally, we compare the cutoff Hamiltonian and channel projection with their counterparts in the full model. Choosing the cutoff large enough and then the total momentum sufficiently large preserves a positive gap on $\sC_P^\perp$. Here, it is crucial that the ground state energies converge uniformly in total momentum as the cutoff is removed; see also \cite{LampartMitrouskasMysliwy} and \eqref{eq:uniform-energy-convergence}.

The remainder of this manuscript is organized as follows. In Section \ref{sec:prelims} we introduce the Fock space and all basic objects that we will work with throughout. We also provide a list of items that we require from (recently) established literature and prove a pull-through formula that will be useful throughout. Section \ref{sec:one_phonon_channel} constructs the states in \eqref{equ:low_energy_states} and verifies basic properties of $\sC_P$. In the following Section \ref{sec:projection_construction} we deal with the problem that $J_P J_P ^*$ is only almost the projection onto $\sC_P$. We then continue in Section \ref{sec:channel_energy_estimates} with estimating $H(P,\alpha)$ on vectors of the form $J_P f$ and potential cross-terms. After that, Section \ref{sec:orthogonal_to_channel} verifies the coercive estimate needed for $H(P,\alpha)$. Finally, we combine all obtained estimates to prove Theorem \ref{thm:main}. Appendix \ref{app:basic_froehlich} collects known properties of the Fr\"ohlich model in our notation and in the precise form we need them.

\section{Preliminaries}
\label{sec:prelims}

\subsection{Fock space and field operators}
\label{subsec:fock-space-preliminaries}

The one-phonon Hilbert space is $\mathfrak h=L^2(\mathbb R^3,\mathrm dk)$, and the bosonic Fock space over $\mathfrak h$ is
\[
  \mathcal F=\Gamma_{\mathrm s}(\mathfrak h)
  =\bigoplus_{n=0}^{\infty}\mathcal F^{(n)},
  \qquad
  \mathcal F^{(n)}=L^2_{\mathrm{sym}}((\mathbb R^3)^n),
  \qquad
  \mathcal F^{(0)}=\mathbb C.
\]
Here $L^2_{\mathrm{sym}}((\mathbb R^3)^n)$ consists of square-integrable functions invariant under permutations of their arguments. Thus, a vector $\xi\in\mathcal F$ is a sequence $(\xi^{(n)})_{n\geq0}$ with
\[
  \|\xi\|_{\mathcal F}^2
  =\sum_{n=0}^{\infty}\|\xi^{(n)}\|_{\mathcal F^{(n)}}^2<\infty.
\]
We denote the vacuum by $\Omega=(1,0,0,\ldots)$ and the dense subspace of vectors with only finitely many nonzero components by
$\mathcal F_{\mathrm{fin}}$. For a real-valued measurable function $\omega$ on $\mathbb R^3$, its second quantization $\mathrm d\Gamma(\omega)$ is the self-adjoint operator given by
\[
  \bigl(\mathrm d\Gamma(\omega)\xi\bigr)^{(n)}(k_1,\ldots,k_n)
  =\left(\sum_{j=1}^n\omega(k_j)\right)
    \xi^{(n)}(k_1,\ldots,k_n),
\]
with its maximal domain in $\mathcal F$; the sum is zero on the vacuum sector. In particular, the phonon number and momentum operators are
\[
  \mathcal N=\mathrm d\Gamma(1),
  \qquad
  P_{\mathrm f}=\mathrm d\Gamma(k),
\]
where the latter is understood component-wise. On the $n$-phonon sector, these operators act as multiplication by $n$ and $k_1+\cdots+k_n$, respectively. For $f\in\mathfrak h$, the annihilation and creation operators are initially defined on $\mathcal F_{\mathrm{fin}}$ by
\[
  \bigl(a(f)\xi\bigr)^{(n)}(k_1,\ldots,k_n)
  =\sqrt{n+1}\int_{\mathbb R^3}\overline{f(k)}\,
    \xi^{(n+1)}(k,k_1,\ldots,k_n)\,\mathrm dk,
  \qquad n\geq0,
\]
and
\[
  \bigl(a^*(f)\xi\bigr)^{(n)}(k_1,\ldots,k_n)
  =\frac{1}{\sqrt n}\sum_{j=1}^n f(k_j)
    \xi^{(n-1)}(k_1,\ldots,\widehat{k_j},\ldots,k_n),
  \qquad n\geq1,
\]
with $(a^*(f)\xi)^{(0)}=0$. The hat in the previous display denotes an omitted variable. In what follows, we may also write $\mathbf{k}/k_j$ to denote the omitted variable $k_j$. We use the same notation for the closures of creation and annihilation operators, which satisfy $a^*(f)=a(f)^*$. With our convention for inner products, $a(f)$ is anti-linear and $a^*(f)$ is linear in $f$. On $\mathcal F_{\mathrm{fin}}$, the canonical commutation relations read
\[
  [a(f_1),a^*(f_2)]=\langle f_1,f_2\rangle\,\mathbf 1,
  \qquad
  [a(f_1),a(f_2)]=[a^*(f_1),a^*(f_2)]=0,
  \qquad f_1,f_2\in\mathfrak h.
\]
Both operators are defined on $\mathcal D(\mathcal N^{1/2})$, and
\[
  \|a(f)\xi\|\leq\|f\|\,
    \|\mathcal N^{1/2}\xi\|,
  \qquad
  \|a^*(f)\xi\|\leq\|f\|\,
    \|(\mathcal N+1)^{1/2}\xi\|.
\]

The field operator $\Phi(f)$ is the self-adjoint closure of $a(f)+a^*(f)$ on $\mathcal F_{\mathrm{fin}}$. In particular, for $\xi\in\mathcal D(\mathcal N^{1/2})$ and $\varepsilon>0$,
\[
  |\langle\xi,\Phi(f)\xi\rangle|
  \leq 2\|f\|\,
    \|\mathcal N^{1/2}\xi\|\,\|\xi\|
  \leq\varepsilon\|\mathcal N^{1/2}\xi\|^2
    +\varepsilon^{-1}\|f\|^2\|\xi\|^2.
\]

\subsection{Nomenclature, fixed constants and required inputs}
Let us  now fix some conventions that will be used in all that follows. We fix $\alpha>0$ throughout most of the proof and drop it from the notation, meaning
\[
 H(P)=H(P,\alpha),\qquad E(P)=E(P,\alpha),\qquad E_0=E(0,\alpha),\qquad E_*=E_0+1, \qquad g= g_\alpha.
\]

We write $B_r (x)=\{q\in\R^3:|q-x|<r\}$, and $M_h$ for multiplication by a function $h$ on $L^2(B_r)$.  Inner products are anti-linear in the first variable and linear in the second.  If $A$ is semibounded, $A[x]$ denotes its closed quadratic form evaluated at $x$.  Operator inequalities involving unbounded operators are understood in the quadratic-form sense unless an operator-domain statement is explicitly made. Whenever a norm does not have a subscript, it is naturally specified from context.
For $\Lambda<\infty$, put
\[
 g_\Lambda=g\one_{\{|k|\le \Lambda\}},
 \qquad
 H_\Lambda(P)=K(P)+\Phi(g_\Lambda),
 \qquad
 K(P)=\frac12(P-\Pf)^2+\N.
\]
The common form domain of the free fiber Hamiltonians $K(P)$ introduced above can be written as
\[
  \mathcal Q=\mathcal D(|P_{\mathrm f}|)
    \cap\mathcal D(\mathcal N^{1/2}),
  \qquad
  K(P)[\xi]=\frac12\|(P-P_{\mathrm f})\xi\|^2
    +\|\mathcal N^{1/2}\xi\|^2.
\]
These definitions apply directly to the cutoff interactions $\Phi(g_\Lambda)$, since $g_\Lambda\in\mathfrak h$. For the full Fr\"ohlich coupling, $g\notin\mathfrak h$ and the Hamiltonian is obtained as the limit of cutoff quadratic forms; see Appendix~\ref{app:basic_froehlich} for partial results and \cite{BrooksEffectiveMass} for a more detailed account.

Our proof uses the following established facts.
\begin{enumerate}[label=\textbf{(I\arabic*)},leftmargin=3.5em]
\item\label{I:hvz} For every $P\in\R^3$,
\[
 \essspec H(P)=[E(0)+1,\infty).
\]
The full interval statement is Lemma~1.1 (i) in \cite{LampartMitrouskasMysliwy}, after the kinetic normalization change described below. 

\item\label{I:geometry} The function $E(P)$ is radial and non-decreasing in $|P|$, has a strict global minimum at $P=0$, and is strictly increasing as long as $E(P)<E_*$.  The strict global minimum is Theorem~1.2 in \cite{LampartMitrouskasMysliwy}; monotonicity and strict monotonicity below threshold are Theorem~1(i) in \cite{Polzer2023}.

\item\label{I:local} $\cI_\alpha$ contains a neighborhood of zero, and on that set $E(P)$ is a simple isolated eigenvalue with real-analytic eigenprojection and energy.  These statements are collected in Lemma~4.1(2), (3), and (6) of \cite{DybalskiSpohn}.  Moreover,
\[
 E(q)-E(0)=\frac{|q|^2}{2m_{\mathrm{eff}}}+o(|q|^2),
 \qquad 1<m_{\mathrm{eff}}<\infty,
\]
where the all-coupling finiteness and strict positivity are Corollary~7 in \cite{Polzer2023}.

\item\label{I:cutoff} Sharp ultraviolet cutoffs converge in norm-resolvent sense uniformly in the fiber momentum.  This is Proposition~A.2 in \cite{LampartMitrouskasMysliwy}.

\item\label{I:gross} For every regular cutoff model, the fiber energy is minimized at momentum zero.  In the present kinetic normalization this is item~0 of Lemma~3.1 in \cite{DybalskiSpohn}. 
\end{enumerate}

\begin{remark}\label{rem:normalization}
In ~\cite{LampartMitrouskasMysliwy}, the operator $(P-P_{\mathrm f})^2$ is used, whereas \eqref{eq:intro-H} and \cite{DybalskiSpohn,Polzer2023} use $\frac12(P-P_{\mathrm f})^2$.  One can map one convention to the other by a scaling argument.
Because all quoted assertions hold for every positive coupling, their qualitative content is invariant under this conversion.  More importantly, all ultraviolet estimates actually used below are reproved in our convention in Proposition~\ref{prop:form-comparison}.
\end{remark}

We now explain carefully how the uniform local constants are chosen.  Let $p_q$ denote the rank-one ground-state projection.  First decrease $r>0$ so
that $\overline{B_{2r}}$ lies in $\cI_\alpha$ and $\|p_q-p_0\|<1/2$ there. The formula
\[
 \psi_q=\frac{p_q\psi_0}{\|p_q\psi_0\|}
\]
then gives a norm continuous map in $q$ on the whole ball. The quadratic expansion in input~\ref{I:local} and the positivity of $m_{\mathrm{eff}}$ give a constant $c>0$ after decreasing $r$ once in the sense that we may choose $r,c,\gamma>0$ such that
\begin{align}
 h(q)=E(q)-E(0)&\ge c|q|^2,\nonumber\\
 H(q)-E(q)&\ge \gamma\bigl(1-|\psi_q\rangle\langle\psi_q|\bigr),\nonumber\\
 \sup_{|q|\le2r}h(q)&\le\frac12.\nonumber
\end{align}
The third bound follows from continuity and $h(0)=0$.  For the second, every $q\in\overline{B_{2r}}$ has a strictly positive gap above its simple isolated ground state. Norm-resolvent continuity prevents this gap from closing under a small change of $q$, which gives the uniform constant $\gamma>0$.

\subsection{A pull-through formula}

For $\xi \in \mathcal D(\mathcal N^{1/2})$, we define the annihilation field component-wise by
\begin{equation}\label{eq:annihilation_components}
    (a(k)\xi)^{(m)}(k_1,\ldots,k_m)
    =
    \sqrt{m+1}\,
    \xi^{(m+1)}(k,k_1,\ldots,k_m),
    \qquad m\ge0.
\end{equation}
Here $m$ is the number of particles remaining after annihilation. These formulas define an element of $L^2(\mathbb R^3;\mathcal F)$, and
\[
    \int_{\mathbb R^3} \de k\,
    \|a(k)\xi\|^2
    =
    \|\mathcal N^{1/2}\xi\|^2.
\]
The field $a(k)\xi$ is therefore understood for almost every $k$. For a bounded measurable set $K\subset\R^3$, put
\[
 \N_K=\mathrm d\Gamma(\one_K)
\]
and define the localized annihilation operator by
\[
 A_K\xi=(k\mapsto\one_K(k)a(k)\xi),
 \qquad
 \cD(A_K)=\cD(\N_K^{1/2}).
\]
The goal of this subsection is to prove the following Lemma.
\begin{lemma}[Pull-through formula]\label{lemma:pull}
For every $|q|\le2r$,
\begin{equation}\label{eq:ptr}
 a(k)\psi_q
 =-g(k)\bigl(H(q-k)+1-E(q)\bigr)^{-1}\psi_q
\end{equation}
locally in $L^2(\R^3;\F)$ and hence for almost every $k$.  For every bounded $K\subset\R^3$,
\begin{equation}\label{eq:graphconv}
 \sup_{|q|\le2r}\norm{A_K(\psi_{\Lambda,q}-\psi_q)}_{L^2(K;\F)}\longrightarrow0.
\end{equation}
Moreover,
\begin{equation}\label{eq:ptrbound}
 \norm{a(k)\psi_q}\le2|g(k)|
\end{equation}
for almost every $k$ and every $|q|\le2r$.
\end{lemma}

We note that commutation relations of $a(k)$ and $H(P)$ are immediate for the regularized Fr\"ohlich Hamiltonians; see Appendix \ref{app:basic_froehlich}. All of the effort made in this subsection is to remove the cutoff and make sure that inverting operators is well-defined.

\begin{proposition}\label{lem:riesz-local}
For all sufficiently large $\Lambda$ and every $|q|\le2r$, the cutoff Hamiltonian $H_\Lambda(q)$ has a simple ground-state eigenvalue $E_\Lambda(q)$ separated from the rest of the spectrum by at least $\gamma/2>0$.  Its rank-one projection $p_{\Lambda,q}$ can be chosen so that
\begin{equation}\nonumber
 \sup_{|q|\le2r}\norm{p_{\Lambda,q}-p_q}\longrightarrow0.
\end{equation}
Moreover, writing $\psi_{\Lambda,q}=\frac{p_{\Lambda,q}\psi_q}{\norm{p_{\Lambda,q}\psi_q}}$, one has
\begin{equation}\nonumber
 \delta_\Lambda = \sup_{|q|\le2r}\norm{\psi_{\Lambda,q}-\psi_q}\longrightarrow0.
\end{equation}
The maps $q\mapsto p_{\Lambda,q}$ and $q\mapsto\psi_{\Lambda,q}$ are norm-continuous on $\overline{B_{2r}}$.
\end{proposition}

\begin{proof}
    Recall that $\gamma$ was chosen as the uniform spectral gap of $H(q)$ for $|q| \le 2 r$. Fix $q \in B_{2r}(0)$. Let  $\Gamma_q \subset \rho(H(q))$ be the spherical contour with radius $\gamma3/4$ centred at $E(q)$. By definition of the resolvent we then have the simple bound
    \begin{equation}
        \label{eq:physical_countour}
        \sup_{q \in B_{2r}} \sup_{z \in \Gamma_q} \Vert (H(q)  - z)^{-1} \Vert \le \frac 4 \gamma.
    \end{equation}
    By Proposition \ref{prop:uniform_resolvent_conv} we have for $z$ not in the spectrum of any $H(q)$ that
    $$\sup_{q \in B_{2r}} \Vert (H_\Lambda (q) - z)^{-1} - (H(q) - z)^{-1} \Vert \to 0.$$
    This  yields via the resolvent identity the uniform convergence on the compact set $\Gamma_q$. Indeed, choose $b>C_0 +1$, where $-C_0$ is the uniform lower bound for all $H_\Lambda (P)$ with $\Lambda$ sufficiently large and $P \in \bbR^3$ (cf. Proposition \ref{prop:form-comparison}). Then all $(H_\Lambda(P)+b)^{-1}$ are bounded operators with norm $\le 1$. Abbreviating $R_\Lambda(z,q) = (H_\Lambda(q)-z)^{-1}$ with $R(z,q) = (H(q)-z)^{-1}$ and $D_\Lambda(z,q) = (H_\Lambda (q)-z) (H_\Lambda(q)+b)^{-1}$  as well as $D(z,q) =(H(q)-z) (H(q)+b)^{-1}$, we have (for now)  the formal identity
    \begin{equation}
        \label{eq:invert_calc}
        R_\Lambda(z,q) - R(z,q)= D_\Lambda (z,q) ^{-1}\big (R_\Lambda(-b,q) - R(-b,q)\big)D(z,q)^{-1}.
    \end{equation}
    This is verified by calculating  $D_\Lambda(z,q) (R_\Lambda(z,q) - R(z,q))D(z,q) = R_\Lambda(-b,q)-R(-b,q)$ and then multiplying the inverses on both sides. We have that the resolvent at $-b$ converges uniformly as $\Lambda \to \infty$ in $|q|\le 2r$, so it remains to show that $D_\Lambda(z,q)$ is invertible with bounded inverse for $\Lambda$ large enough and $(z,q) \in \Gamma_q \times B_{2r}$. This is seen by writing 
    $$D_\Lambda(z,q) = D(z,q)\big[ I-(z+b)D(z,q)^{-1}\big(R_\Lambda(-b,q)- R(-b,q)\big)\big].$$
    The fact that $D(z,q)$ is invertible is seen by definition. Moreover, by the identity $D(z,q) = 1-(z+b)(H(q)+b)^{-1}$, we also find that $D(z,q)$ has uniformly bounded operator norm. 
    For the inverse we find the formula
    $$D(z,q)^{-1}=1+(z+b)(H(q)-z)^{-1}.$$
    by the contour bound \eqref{eq:physical_countour} found earlier, the above is uniformly bounded in our parameter range. Because $R_\Lambda(-b,q)- R(-b,q)$ converges uniformly in $q$ over compacts, we can bound 
    $$\sup\limits_{|q|\le 2r}\sup\limits_{z \in \Gamma_q}\Vert(z+b)D(z,q)^{-1}\big(R_\Lambda(-b,q)- R(-b,q)\big)\Vert \le \frac 1 2$$
    for $\Lambda$ large enough. The Neumann series then shows that $D_\Lambda (z,q)$ is invertible with uniformly bounded norm in the considered parameter range. This verifies \eqref{eq:invert_calc} 
    and so
    $$\sup_{q \in B_{2r}} \sup_{z \in \Gamma_q} \Vert (H_\Lambda (q) - z)^{-1} - (H(q) - z)^{-1} \Vert \to 0.$$
    The conclusion is that the Riesz projections are converging uniformly, meaning
    $$\sup_{q \in B_{2r}}\Vert p_{\Lambda,q} - p_q \Vert \to 0,$$
    which yields the last two assertions.
\end{proof}

\begin{proposition}\label{lem:AK}
The operator $A_K:\F\supset\cD(A_K)\to L^2(K;\F)$ is closed and
\begin{equation}\nonumber
 \norm{A_K\xi}_{L^2(K;\F)}^2=\Vert\N_K ^{1/2}\xi\Vert^2.
\end{equation}
In particular, convergence in the graph norm of $A_K$ is equivalent to convergence in $\F$ together with convergence of the localized annihilation fields in $L^2(K;\F)$.
\end{proposition}

\begin{proof}
    Direct calculation.
\end{proof}
The following pull-through formula is known at least since the work \cite{Moller}. We state it again for convenience and in the precise setting we need it.
\begin{lemma}\label{lem:cutoff-pull}
For all sufficiently large finite $\Lambda$, every $|q|\le2r$, and almost every $k$,
\begin{equation}\label{eq:ptrR}
 a(k)\psi_{\Lambda,q}
 =-g_\Lambda(k)\bigl(H_\Lambda(q-k)+1-E_\Lambda(q)\bigr)^{-1}\psi_{\Lambda,q}.
\end{equation}
For all sufficiently large $\Lambda$,
\begin{equation}\label{eq:denom}
 H_\Lambda(q-k)+1-E_\Lambda(q)\ge\frac13
 \qquad(|q|\le2r,\ k\in\R^3).
\end{equation}
\end{lemma}

\begin{proof}
As proven in Appendix \ref{app:basic_froehlich}, it holds that
$$a(k) H_\Lambda (q) = ( H_\Lambda (q-k) +1) a(k)+g_\Lambda (k).$$
on the operator core and for the respective ground states.
By \ref{I:gross} we have $H_\Lambda (q-k) \ge E_\Lambda (0)$. Because $E_\Lambda (q)$ converges uniformly over all $q$,
$$H_\Lambda (q-k)+1 - E_\Lambda(q) \ge E_\Lambda (0) +1 - E_\Lambda(q) \ge \frac 1 3$$
for $\Lambda$ sufficiently large because $|q|\le2r$. In particular, the operator on the left-hand side is invertible. The proof then follows via Lemma \ref{lem:finite_cutoff_commutator_gs} by plugging in the ground state and solving for $a(k) \psi_{\Lambda,q}$.
\end{proof}

In order to conclude that the cutoff can be removed, the following assertion is of help.

\begin{lemma}\label{lem:shifted}
Let $K_q\subset \overline{B_{2r}}$, $K_k \subset \bbR^3$ be compact. Then
\[
 \sup_{q\in K_q,\,k\in K_k}\norm{(H_\Lambda(q-k)+1-E_\Lambda(q))^{-1}-(H(q-k)+1-E(q))^{-1}}\longrightarrow0.
\]
\end{lemma}

\begin{proof}
    First, note that  $H_\Lambda(q-k)+1-E_\Lambda(q)$ has spectrum bounded below by $1/3$ by Lemma \ref{lem:cutoff-pull}. The same bound holds for $H(q-k)+1-E(q)$. Therefore, both inverses exist and are well-defined via the functional calculus. 
    The claim follows from the resolvent identity, \eqref{eq:uniform-energy-convergence}, Proposition \ref{prop:uniform_resolvent_conv} and the uniform boundedness of the inverses.
\end{proof}

\begin{proof}[Proof of Lemma \ref{lemma:pull}]
Thanks to Lemma \ref{lem:cutoff-pull} the pull-through formula is already verified for the model with cutoff. The convergence of $a(k) \psi_{\Lambda,q} \to a(k) \psi_q$ follows from Proposition \ref{lem:riesz-local} and Proposition \ref{lem:AK}. We show that the left and right sides both converge as functions in $L^2(K;\cF)$, for arbitrary $K \subseteq \bbR^3$ compact. Indeed, we first pass to the limit on the right hand side, and then use closeness to obtain the desired convergence. Because
$$\bigl(H_\Lambda(q-k)+1-E_\Lambda(q)\bigr)^{-1} \to \bigl(H(q-k)+1-E(q)\bigr)^{-1}$$
 uniformly in operator norm, 
$$\Vert \bigl(H_\Lambda(q-k)+1-E_\Lambda(q)\bigr)^{-1} \psi_{\Lambda,q} - \bigl(H(q-k)+1-E(q)\bigr)^{-1} \psi_q \Vert \to 0.$$
Multiplying this by $g_\Lambda$, which is just $g$ on $K$ and integrating on $K$ shows
$$g_\Lambda(k)\bigl(H_\Lambda(q-k)+1-E_\Lambda(q)\bigr)^{-1} \psi_{\Lambda,q} \to g(k)\bigl(H(q-k)+1-E(q)\bigr)^{-1} \psi_q $$
in $L^2(K;\cF)$.
This implies the pull-through formula \eqref{eq:ptr} without cutoff thanks to the convergence of ground states and closeness of $A_K$ for almost every $k \in \bbR^3$.
\eqref{eq:ptrbound} is then a direct consequence of the pull-through formula by noting that the inverse operator on the right-hand side is uniformly bounded (by choice of $|q|\le 2r$). Finally, the convergence 
$$\sup_{|q|\le2r}\norm{A_K(\psi_{\Lambda,q}-\psi_q)}_{L^2(K;\F)}\longrightarrow0$$
also follows from applying the pull-through formula to $A_K (\psi_{\Lambda,q}- \psi_q)$, which yields
\begin{equation}
    \nonumber
    \begin{split}
        &\sup\limits_{|q|\le 2r} \Vert A_K (\psi_{\Lambda,q}- \psi_q)\Vert_{L^2 (K; \cF)} ^2 = \sup\limits_{|q|\le 2r} \int_K \de k \, \big\Vert a(k) \psi_{\Lambda,q}- a(k) \psi_q \big\Vert^2   \\
        &\le \sup\limits_{|q|\le 2r} \int_K \de k \, g(k)^2 \big\Vert \bigl(H_\Lambda(q-k)+1-E_\Lambda(q)\bigr)^{-1} \psi_{\Lambda,q} - \bigl(H(q-k)+1-E(q)\bigr)^{-1} \psi_q \big\Vert^2. 
    \end{split}
\end{equation}
The right-hand side tends to $0$ thanks to Lemma \ref{lem:shifted} and the fact that 
$\int_K \de k \, g(k)^2 < \infty$
in three dimensions. This verifies \eqref{eq:graphconv} and thus finishes the proof.
\end{proof}

\section{The space $\sC_P$ and basic properties}
\label{sec:one_phonon_channel}
The goal of this section is to rigorously construct states of the form \eqref{equ:low_energy_states}.
As before, it will be useful to think of $k \to 1_K a(k)\psi$ for some compact $K \subseteq \R^3$ as an element in $L^2(K;\cF)$. Because the construction of our desired low-energy states is slightly unconventional, we first discuss measurability issues.
\begin{lemma}\label{lem:measurable}
Let $B$ be a standard Borel space and let $B \ni q\mapsto\eta_q\in\F$ be strongly measurable.  The Fock components may be represented by jointly measurable functions
\[
 (q,k_1,\ldots,k_n)\longmapsto\eta_q^{(n)}(k_1,\ldots,k_n).
\]
If, in addition, a locally square-integrable annihilation field is given as a strongly measurable map $(q,k)\mapsto F(q,k)\in\F$, then the component representatives may be chosen so that $F(q,k)=a(k)\eta_q$ almost everywhere on every bounded momentum set.
\end{lemma}

\begin{proof}
    Fix $n \in \mathbb N$. Since $\cF^{(n)}$ is separable, choose a dense, countable set $\{\nu_j:j \in \mathbb N\}$ and fix a measurable representative for each of them. Abbreviate a momentum vector by $\mathbf{x}=(x_1,...,x_n) \in \mathbb R ^{3n}$. Let now
    $$j_m (q) = \min \left\{ j : \Vert \nu_j - \eta_q ^{(n)} \Vert < 2^{-m} \right\}, \qquad s_m (q,\mathbf x) =\nu_{j_m(q)}(\mathbf x).$$
    Note that $s_m$ is measurable in the first variable by the assumption that $q \to \eta_q$ is strongly measurable. Clearly $s_m(q,\cdot)$ is also measurable. Joint measurability follows, as $j_m \in \mathbb N$. For fixed $q$,
    $$\sum\limits_{m=1}^\infty \Vert s_{m+1}(q,\cdot) - s_m (q,\cdot) \Vert < \infty,$$
    meaning that $s_m (q,\cdot) \to \eta^{(n)}_q$ in $L^2(\R ^{3n})$. By the triangle inequality and monotone convergence we find
    $$\sum\limits_{m=1}^\infty | s_{m+1}(q,x) - s_m (q,x) | < \infty$$
    for almost every $x$.
    Therefore, $v(q,x) =\lim\limits_{m\to \infty }s_m(q,x)$ exists pointwise for almost every $x$ (we set it to $0$ on a possible zero set). Note that almost everywhere $v(q,\cdot) = \eta_q ^{(n)}$; this is best seen by Fatou's Lemma, because
    $$\int \de x \, |v(q,x) - \eta^{(n)}_q (x)|^2 \le \liminf_{m \to \infty} \int \de x \, |s_m(q,x)  - \eta^{(n)}_q (x)|^2 = 0.$$
    Finally, $v$ is jointly measurable as limit of jointly measurable functions, as desired. The statement about annihilation fields follows directly from what we have just shown and Fubini's Theorem.
\end{proof}
Suppose $q\mapsto\eta_q\in\F$ is measurable and the relevant localized annihilation fields possess a jointly measurable representative as in Lemma~\ref{lem:measurable}. Take $f: B_r \to \bbC$ measurable. Define for $n \ge  0$ the vector $C^\eta f$ by
\begin{align}\label{eq:component}
 (C^\eta f)^{(n+1)}(k_1,\dots,k_{n+1})
 =\frac1{\sqrt{n+1}}\sum_{j=1}^{n+1}
 &\one_{B_r (P)}(k_j)f(P-k_j)
 \eta_{P-k_j}^{(n)}(k_1,\dots,\widehat{k_j},\dots,k_{n+1}).
\end{align}
Set also $(C^\eta f)^0 = 0$.
We will use the above definition mainly for the family of ground states $\psi = \{\psi_q: q \in B_r\}$.
The reason we define it for general families $\eta = \{\eta_q : q \in B_r\}$ is that we will also work with cut-off Hamiltonians, in which case $\psi$ must be replaced by the ground states of the corresponding cutoff Hamiltonians. For a measurable $f$ define the two non-negative quantities
\begin{align*}
 D_\eta(f)
 &=\int_{B_r} \dd q \, |f(q)|^2\|\eta_q\|^2,\\
 X_\eta(f)
 &=\iint_{B_r\times B_r}\dd q\,\dd q' \, |f(q)|\,|f(q')|
 \|a(P-q')\eta_q\|\,\|a(P-q)\eta_{q'}\|.
\end{align*}
The second quantity is an absolute-integrability condition for the exchange term below.  It is finite in every application in this paper. By Lemma \ref{lem:measurable}, the states $C^\eta f$ are all measurable on every Fock sector.

\begin{lemma}\label{lem:creation}
If $D_\eta(f)+X_\eta(f)<\infty$, then \eqref{eq:component} defines a vector of $\F$ and
\begin{align}\label{eq:creationnorm}
 \norm{C^\eta f}^2
 &=\int_{B_r} \dd q \, |f(q)|^2\norm{\eta_q}^2
 +\iint_{B_r\times B_r}\dd q\,\dd q' \, \overline{f(q)}f(q')
 \ip{a(P-q')\eta_q}{a(P-q)\eta_{q'}}.
\end{align}
Moreover, in case the right-hand side is finite, we also have the difference bound
\begin{equation}\label{eq:creationdiff}
 \norm{(C^\eta-C^{\widetilde\eta})f}^2
 \le\left(
 \sup_{q\in B_r}\norm{\eta_q-\widetilde\eta_q}^2
 +\sup_{q\in B_r}\int_{B_r}\dd q'\,\norm{a(P-q')(\eta_q-\widetilde\eta_q)}^2
 \right)\norm f^2.
\end{equation}
\end{lemma}

\begin{proof}
    Fix $f $ measurable with $D_\eta(f)+X_\eta(f)<\infty$. We start on the $n$-th Fock sector. Abbreviate $\mathbf k = (k_1,...,k_n) \in \bbR^{3n}$ and calculate
    \begin{equation}
        \label{eq:channel_inner_product_first_calculation}
        \begin{split}
            &\langle (C^\eta f)^{(n)}, (C^\eta f)^{(n)}\rangle \\
            &= \frac 1 n \int\de \mathbf{k}\, \sum\limits_{i,j=1}^n 1_{B_r(P)}(k_i)1_{B_r(P)}(k_j)\overline{f(P-k_i)\eta_{P-k_i}^{(n-1)}(\mathbf{k}\backslash k_i})f(P-k_j) \eta_{P-k_j}^{(n-1)}(\mathbf{k}\backslash k_j).
        \end{split}
    \end{equation}
    Here, $\mathbf{k}\backslash k_i = (k_1,\dots,\hat{k}_i,\dots,k_n)$, where a hat indicates an omitted index.
    Let us now separate two cases. If $i = j$, the terms become
    $$\int_{B_r} \de q \, |f(q)|^2 \big\Vert \eta_{q}^{(n-1)} \big\Vert ^2.$$
    Because there are $n$ of them, the prefactor $1/n$ in \eqref{eq:channel_inner_product_first_calculation} cancels, giving the first term in \eqref{eq:creationnorm}. The cross-terms evaluate after a change of variables $u_j = P-k_j$ to 
    $$(n-1)
    \iint_{B_r \times B_r} \de q \de q' \overline{f(q)} f(q') \langle \eta^{(n-1)}_{q}(P-q',\cdot), \eta^{(n-1)}_{q'} (P-q,\cdot) \rangle.$$ 
    But $\sqrt{n-1}\eta_q ^{(n-1)} (P-q',\cdot) = a(P-q')\eta_q ^{(n-1)}$.
    Summing over $n$ and exchanging summation and integration by dominated convergence (recall that we asked for $D_\eta (f) + X_\eta (f) < \infty$) yields \eqref{eq:creationnorm}. Next, note that $C^{\kappa}- C^{\theta} = C^{\kappa-\theta}$ for two families $\kappa,\theta$. We can therefore extract \eqref{eq:creationdiff} from \eqref{eq:creationnorm} by considering $\eta = \kappa - \theta$ and bounding both terms on the right side.
    For the first term, we bound
    $$\int_{B_r} \dd q \, |f(q)|^2\norm{\eta_q}^2 \le \sup_{q \in B_r} \Vert \eta_q\Vert^2 \Vert f \Vert^2.$$
    For the second term, use $xy \le \frac 12  (x^2+y^2)$ to bound
    \begin{equation}
        \nonumber
        \begin{split}
            \left| \overline{f(q)}f(q')
 \ip{a(P-q')\eta_q}{a(P-q)\eta_{q'}} \right| &\le \big|\overline{f(q)}\big| \, \Vert a(P-q')\eta_{q} \Vert \, \big|f(q')\big| \, \Vert a(P-q)\eta_{q'} \Vert \\
 &\le \frac 1 2 \left( |f(q)|^2 \Vert a(P-q')\eta_{q} \Vert^2 + |f(q')|^2 \Vert a(P-q)\eta_{q'} \Vert^2  \right).
        \end{split}
    \end{equation}
Integrating this over $B_r \times B_r$ yields the bound
$$X_\eta (f)  \le \frac 1 2 \iint_{B_r \times B_r}\de q \de q' \,  |f(q)|^2 \Vert a(P-q')\eta_{q} \Vert^2 + |f(q')|^2 \Vert a(P-q)\eta_{q'} \Vert^2.$$
Noting that
$$\int_{B_r}\de q \, |f(q)|^2  \int_{B_r }\de q' \, \Vert a(P-q')\eta_{q} \Vert^2 \le \int_{B_r} \de q \, |f(q)|^2  \sup\limits_{q'' \in B_r}\int_{B_r }\de q' \,   \Vert a(P-q')\eta_{q''} \Vert^2$$
finishes the proof.
\end{proof}

For the physical family, write $J_P=C^\psi$ and use the suggestive notation
\begin{equation}\label{eq:J-and-S}
 J_Pf=\int_{B_r} \dd q \, f(q)a^*(P-q)\psi_q,
 \qquad
 S_Pf=\int_{B_r} \dd q \, f(q)g(P-q)\psi_q.
\end{equation}
The first integral is simply a shorthand for the component definition \eqref{eq:component}. In contrast, the second integral is a proper Bochner integral (the integrand is strongly measurable and integrable by Cauchy-Schwarz). Both maps are bounded as maps from $L^2(B_r)$ into Fock space.  Indeed, with $b_P(q)=|g(P-q)|$, the pull-through estimate and Lemma~\ref{lem:creation} give
\[
 \norm{J_Pf}^2
 \le \bigl(1+4\norm{b_P}_{L^2(B_r)}^2\bigr)\norm f^2,
 \qquad
 \norm{S_Pf}\le\norm{b_P}_{L^2(B_r)}\norm f.
\]
For a cutoff $\Lambda>|P|+r$, define $J_{\Lambda,P}=C^{\psi_\Lambda}$ and define $S_{\Lambda,P}$ by replacing $\psi_q$ with $\psi_{\Lambda,q}$ in \eqref{eq:J-and-S}.  Then
$g_\Lambda(P-q)=g(P-q)$ on $B_r$.  For all sufficiently large $\Lambda$, \eqref{eq:ptrR} and \eqref{eq:denom} give
\[
 \norm{a(P-q')\psi_{\Lambda,q}}\le3b_P(q')
 \qquad(q,q'\in B_r).
\]
Consequently Lemma~\ref{lem:creation} shows that $J_{\Lambda,P}$ and $S_{\Lambda,P}$ are bounded, uniformly in all sufficiently large $\Lambda$ for each fixed $P$.

\begin{lemma}\label{lem:cutoff-channel}
For every $f\in L^2(B_r)$,
\begin{equation}\nonumber
 \bigl(H_\Lambda(P)-E_\Lambda(0)-1\bigr)J_{\Lambda,P}f
 =J_{\Lambda,P}M_{h_\Lambda}f+S_{\Lambda,P}f
\end{equation}
with $h_\Lambda(q)=E_\Lambda(q)-E_\Lambda(0)$. In particular, $J_{\Lambda,P}f\in\cD(H_\Lambda(P))$.
\end{lemma}

\begin{proof}
    Let $\mathfrak C$ be the operator core from Lemma \ref{lem:finite_cutoff_core}. Let $\Xi \in \mathfrak C$. Moreover, take $f \in L^2(B_r)$. Then, by the component-wise definition of $J_{\Lambda,P}$ and bosonic symmetry,
    \begin{equation}
        \label{eq:simple_J_action}
        \langle \Xi, J_{\Lambda,P} f \rangle = \int_{B_r} \de q \, f(q) \langle a(P-q) \Xi, \psi_{\Lambda,q} \rangle.
    \end{equation}
    \eqref{eq:simple_J_action} is directly extended to $\Xi \in \cD(\cN^{1/2})$ because $\mathfrak C$ is a core for $\cN^{1/2}$ as well.
    Using now Lemma \ref{lem:finite_cutoff_commutator} , we find
    $$a(P-q)H_{\Lambda}(P) \xi = \left[ (H_\Lambda(q) +1)a(P-q) + g_\Lambda (P-q) \right] \xi.$$
    Inserting this into \eqref{eq:simple_J_action} with the choice $\Xi = H_\Lambda (P) \xi$ for $\xi \in \mathfrak C$ together with the self-adjointness of $H_{\Lambda}(q)$ and Lemma \ref{lem:finite_cutoff_commutator_gs} yields
    \begin{equation}
        \nonumber
        \begin{split}
            \langle H_\Lambda (P)\xi, J_{\Lambda,P} f \rangle &= \int_{B_r} \de q \, f(q) \langle a(P-q) \xi , (E_{\Lambda}(q)+1) \psi_{\Lambda,q} \rangle + \int_{B_r} \de q \, f(q) g_\Lambda(P-q)\langle \xi, \psi_{\Lambda,q}    \rangle \\
            &= \langle \xi, J_{\Lambda,P} M_{E_\Lambda(\cdot)+1} f \rangle + \langle \xi, S_{\Lambda,P}f \rangle. 
        \end{split}
    \end{equation}
    The second equality follows by the definition of $S_{\Lambda,P}f$ and once again applying \eqref{eq:simple_J_action}. Subtracting $E_\Lambda(0)+1$ from the Hamiltonian in the last display yields
    $$\langle [H_\Lambda (P) - E_\Lambda (0)-1] \xi , J_{\Lambda,P}f \rangle = \langle \xi , J_{\Lambda,P} M_{h_\Lambda}f + S_{\Lambda,P}f \rangle. $$
    Since $\mathfrak C$ is an operator core and by self-adjointness of $H_\Lambda (P)$, this shows the desired commutation relations and $J_{\Lambda,P}f \in \cD(H_\Lambda(P))$.
\end{proof}

\begin{proposition}\label{prop:intertwine}
The map $J_P:L^2(B_r)\to\F$ takes values in $\cD(H(P))$, is bounded for the graph norm of $H(P)$, and satisfies
\begin{equation}\label{eq:intertwine}
 \bigl(H(P)-E_*\bigr)J_P=J_PM_h+S_P.
\end{equation}
\end{proposition}

\begin{proof}
    We begin with showing $J_{\Lambda,P}M_{h_\Lambda}+ S_{\Lambda,P} \to J_P M_h +S_P$ in operator norm. First, note
    $$\Vert (S_{\Lambda,P}-S_P)f \Vert \le \int_{B_r} \de q |f(q)|b_P(q) \Vert \psi_q - \psi_{\Lambda,q}\Vert \le \delta_\Lambda \Vert b_P \Vert \, \Vert f \Vert.$$
    This verifies $S_{\Lambda,P} \to S_P$ in operator norm. Next, \eqref{eq:creationdiff} implies together with Proposition \ref{lem:riesz-local} and \eqref{eq:graphconv} that $J_{\Lambda,P} \to J_P$ in operator norm. Now, because $ \sup\limits_{|q|\le 2r} |E_\Lambda (q) -E(q)| \to 0$ thanks to Proposition \ref{prop:form-comparison}, we have $\Vert h_\Lambda - h \Vert _\infty \to 0$. This implies the claimed operator norm convergence. Indeed,
    $$\Vert J_{\Lambda,P}M_{h_\Lambda}-J_P M_h \Vert \le \Vert J_{\Lambda,P} - J_P \Vert \, \Vert h_\Lambda \Vert_\infty + \Vert J_P \Vert \, \Vert h_\Lambda - h \Vert_\infty \to 0.$$
    It remains to verify 
    $$
    \bigl(H_\Lambda(P)-E_\Lambda(0)-1\bigr)J_{\Lambda,P}f \to 
    \bigl(H(P)-E(0)-1\bigr)J_{P}f.
    $$
    Analogously to $E_* = E(0)+1$, define $E_{\Lambda,*} = E_\Lambda (0) +1$.
    By the resolvent identity it holds
    $$\Vert (H_\Lambda(P) - E_{\Lambda,*}-i)^{-1} - (H_\Lambda(P) - E_{*}-i)^{-1}\Vert \le |E_{\Lambda,*}-E_*|.$$
    But also
    $$\Vert (H_\Lambda(P) - E_{*}-i)^{-1} - (H(P) - E_{*}-i)^{-1}\Vert \to 0$$
    by norm resolvent convergence of $H_\Lambda (P) \to H(P)$. This shows $H_\Lambda(P) - E_\Lambda (0) - 1 \to H(P) - E (0) - 1$ in norm resolvent sense. Fix some $f \in L^2(B_r)$. Then, by Lemma \ref{lem:cutoff-channel},
    $J_{\Lambda,P}f \in \cD(H_\Lambda(P))$. Moreover,
    $$J_{\Lambda,P}f = (H_\Lambda(P)-E_\Lambda(0)-1-i)^{-1} (J_{\Lambda,P}M_{h_\Lambda}f +S_{\Lambda,P}f - iJ_{\Lambda,P}f).$$
    Taking $\Lambda \to \infty$ then shows
    $$J_{P}f = (H(P)-E(0)-1-i)^{-1} (J_{P}M_{h}f +S_{P}f - iJ_{P}f).$$
    Because the range of the inverse of an operator is its domain, $J_{P}f \in \cD(H(P))$ and \eqref{eq:intertwine} follows by multiplying $H(P)-E(0)-1-i$ on both sides and adding $i J_P f$. 
    Finally, by \eqref{eq:intertwine} we have
    $$\Vert H(P)J_P f\Vert \le \big((|E_*| + \Vert h \Vert_\infty) \Vert J_P \Vert + \Vert S_P \Vert \big) \Vert f\Vert.$$
    But then
    $\Vert J_P f \Vert + \Vert H(P)J_P f\Vert \le C \Vert f \Vert,$
    as desired.
\end{proof}

\section{Projecting onto $\sC_P$}
\label{sec:projection_construction}

Let us put $\sC_P=\operatorname{Ran}J_P$, $G_P=J_P^*J_P$ and $b_P(q)=|g(P-q)|$. Recall that $h(q)=E(q)-E(0)$ and $h(q)\ge c|q|^2$ on $B_r$.
For $|P|>r$, define
\begin{equation}\nonumber
 \beta_P^2=\int_{B_r}\de q\,\frac{b_P(q)^2}{h(q)},
\end{equation}
where the value of the integrand at $q=0$ is irrelevant. For future use we record the immediate bound
\begin{equation}
    \beta_P^2\le\frac{2\alpha r}{\pi c(|P|-r)^2}
 =O_\alpha(|P|^{-2})
 \qquad\text{as }|P|\to\infty.\label{eq:projection-beta-bound}
\end{equation}
The goal of this section is to prove the following proposition.

\begin{proposition}\label{prop:projection}
There are constants $C<\infty$ and $P_0<\infty$ such that, for every
$P\in\bbR^3$ with $|P|\ge P_0$, one has $\frac12\le G_P\le\frac32$.
In particular, $G_P$ is invertible and $\sC_P$ is closed.
Put $F_P=G_P^{-1/2}$ and $U_P=J_PF_P$.
Then $U_P:L^2(B_r)\to\cF$ is an isometry onto $\sC_P$, and the
orthogonal projections onto $\sC_P$ and $\sC_P^\perp$ are given by
\begin{equation}\label{eq:projection-formula}
 \Pi_P=U_PU_P^*=J_PG_P^{-1}J_P^*,
 \qquad Q_P=1-\Pi_P.
\end{equation}
Moreover, $\norm{F_P-1}+\norm{\Pi_P-J_PJ_P^*}\le C\beta_P^2$ and,
for every $f\in L^2(B_r)$,
\begin{equation}\label{eq:projection-energy}
 (1-C\beta_P^2)\ip{f}{M_hf}
 \le\ip{F_Pf}{M_hF_Pf}
 \le(1+C\beta_P^2)\ip{f}{M_hf}.
\end{equation}
Finally, $\Pi_P\cF\subset\cD(H(P))$, and both $\Pi_P$ and $Q_P$ leave the form domain $\cQ$ invariant.
\end{proposition}

We first estimate the difference between $J_P^*J_P$ and the identity.

\begin{lemma}\label{lem:projection-kernel}
For every $|P|>r$, one has $G_P=1+R_P$, where $R_P$ is a self-adjoint operator on $L^2(B_r)$ with operator norm bounded by
\begin{align}
 \norm{R_P}\le4\norm{b_P}_{L^2(B_r)}^2
 \le4\norm{h}_{\infty}\beta_P^2.\nonumber
\end{align}
\end{lemma}
\begin{proof}
    Recall Lemma \ref{lem:creation}, which gives for $f \in L^2(B_r)$ that
    $$\norm{J_P f}^2
 = \Vert f \Vert _{L^2 (B_r)} ^2
 +\iint_{B_r\times B_r}\dd q\,\dd q' \, \overline{f(q)}f(q')
 \ip{a(P-q')\psi_q}{a(P-q)\psi_{q'}}.$$
 We  define $R_P$ to be the operator on $L^2(B_r)$ with kernel $k(q,q') = \ip{a(P-q')\psi_q}{a(P-q)\psi_{q'}}$. By Cauchy-Schwarz and pull-through, the kernel is bounded from above by
$$|\ip{a(P-q')\psi_q}{a(P-q)\psi_{q'}}| \le 4b_P(q') b_P (q),$$
which implies that $R_P$ is well-defined, bounded and verifies the claimed operator bound by integrating the bound over $B_r \times B_r$.
\end{proof}

The following lemma shows that $R_P$ and $F_P-1$ are small with respect to the weighted norm $f\mapsto\|M_h^{1/2}f\|$. These estimates will allow us to compare $\langle F_Pf,M_hF_Pf\rangle$ with $\langle f,M_hf\rangle$.

\begin{lemma}\label{lem:projection-energy-error}
There is a constant $C<\infty$ such that, for all sufficiently large
$|P|$ and every $f\in L^2(B_r)$,
\begin{equation}
    \norm{M_h^{1/2}R_Pf}
 +\norm{M_h^{1/2}(F_P-1)f}
 \le C\beta_P^2\norm{M_h^{1/2}f},\label{eq:projection-energy-norm}\\
\end{equation}
\begin{equation}
     \big|\ip{f}{R_PM_hf}\big|
 \le C\beta_P^2\ip{f}{M_hf}.\nonumber
\end{equation}
The constant is independent of $P$ and $f$.
\end{lemma}

\begin{proof}
    Recall that $R_P$ is defined to be the operator with integral kernel
    $$k(q,q') = \ip{a(P-q')\psi_q}{a(P-q)\psi_{q'}}.$$
    Therefore, for any $f \in L^2(B_r)$,
    $$|R_Pf(q)| = \left| \int_{B_r} \de x \, k(q,x) f(x) \right| \le 4 b_P(q) \int_{B_r}\de x \, b_P(x) |f(x)|.$$
    We now bound $\Vert M_h ^{1/2} R_Pf\Vert $. Writing out the inner product explicitly, we find
    \begin{equation}
        \nonumber
        \begin{split}
            \Vert M_h ^{1/2} R_Pf\Vert^2 = \int_{B_r} \de q \, h(q) |R_P f(q)|^2 &\le 16 \int \de q \,  h(q) b_P (q)^2  \left( \int \de q' \,  b_P (q') |f(q')| \right)^2 \\
            &\le 16\Vert M_h ^{1/2} f \Vert^2 \beta_P ^2  \int \de q \,  h(q) b_P (q)^2 .
        \end{split}
    \end{equation}
    Because $0<h(q) \le H_* = \Vert h \Vert_{L^\infty (B_r)}\le 1/2$ almost everywhere,
    $\int \de q \,  h(q) b_P (q)^2 \le H_* ^2 \beta_P ^2,$
    which implies with the previous bound and taking square roots on both sides that
    \begin{equation}
        \label{eq:first_term}
        \Vert M_h ^{1/2} R_Pf\Vert \le H_* \beta_P ^2 4 \Vert M_h ^{1/2} f \Vert.
    \end{equation}
    Next,
    $$\ip{f}{R_PM_hf} = \langle M_h ^{1/2}R_P f,  M_h ^{1/2}f \rangle.$$
    The second claimed inequality therefore follows straight from \eqref{eq:first_term} and Cauchy-Schwarz. We continue with bounding the second term in \eqref{eq:projection-energy-norm}. Choose $P$ large enough so that $H_* \beta_P ^2 4 < \frac 1 8$. Then $F_P = (1+R_P)^{-1/2}$ is a well-defined, bounded operator. In particular,  $\Vert R_P \Vert < 1/8$ due to Lemma \ref{lem:projection-kernel}. This implies that the series expansion
    \begin{equation}
        \label{eq:F_P_series_expansion}
        F_P - 1 = \sum\limits_{n=1}^\infty c_n R_P ^n
    \end{equation}
    converges in operator norm; in the above $c_n = (-1)^n \binom{2n}{n}/4^n$ are the power series coefficients of $(1+z)^{-1/2}$. Iteratively applying \eqref{eq:first_term} with $R_P ^n f$ instead of just $f$ for any $n \in \mathbb N$ yields the bound
    $$\Vert M_h ^{1/2} R_P ^n f \Vert \le (4H_* \beta_P^2)^n \Vert M_h ^{1/2}f \Vert.$$
    Because $M_h ^{1/2}$ is a bounded operator, we may multiply it on both sides of \eqref{eq:F_P_series_expansion} and apply it to each summand. The resulting norm bound yields
    $$\Vert M_h ^{1/2}(F_P - 1)f \Vert \le \sum\limits_{n=1}^\infty |c_n| \, \Vert M_h^{1/2} R_P^nf\Vert \le \Vert M_h ^{1/2} f \Vert \sum\limits_{n=1}^\infty |c_n| (4H_* \beta_P ^2)^n =\Vert M_h ^{1/2} f \Vert ( (1- 4H_* \beta_P ^2)^{-1/2}-1).$$
    Using $(1-x)^{-1/2} - 1 \le 2x$ for $0 < x \le 1/2$, the final bound
    $$\Vert M_h ^{1/2}(F_P - 1)f \Vert \le 8H_* \beta_P ^2\Vert M_h ^{1/2} f \Vert$$
    follows, and the proof is complete.
\end{proof}

\begin{proof}[Proof of Proposition \ref{prop:projection}]
Choose again $|P|$ large enough so that $4H_* \beta_P ^2 < 1/8$. Then, because
$$1-\Vert R_P \Vert  \le  G_P  \le 1+ \Vert R_P \Vert,$$
we find the claimed bound $1/2 \le G_P \le 3/2$. In particular, $F_P = G_P ^{-1/2}$ is well-defined, as before. This directly yields
$$\Vert J_P f \Vert ^2 = \langle f, G_P f \rangle \ge \frac 1 2 \Vert f \Vert^2.$$
In particular, $J_P$ is injective and has closed range. We obtain
$$U_P ^* U_P = F_P J^* _P J_P F_P = F_P G_P F_P = 1,$$
so $U_P$ is isometric. As $F_P$ is invertible, 
$\Ran (U_P) = \sC_P$. This shows that $\Pi_P = U_P U_P^*$ is the projection onto $\sC_P$. Indeed, 
$$\Pi_P ^2 = U_P U_P^* U_P U_P ^* = U_P U_P ^* = \Pi_P.$$
The fact that $\Pi_P$ is self-adjoint and bounded was already verified. Thus, $\Pi_P$ is the orthogonal projection onto $\sC_P$. The formulas \eqref{eq:projection-formula} follow from writing out the definitions and noting that $F_P$ is self-adjoint. The  bound on $\Vert 1-F_P \Vert \le C\beta_P ^2$ is verified exactly as in Lemma \ref{lem:projection-energy-error}, but leaving out the step of multiplying $M_h ^{1/2}$ on both sides. We can also write
$$\Pi_P - J_P J_P ^* = U_P (1-G_P) U_P ^* = - U_P R_P U_P ^* .$$
Because $U_P$ is isometric, $\Vert U_P R_P U_P ^* \Vert = \Vert R_P \Vert$, and so
$$\Vert \Pi_P - J_P J_P ^* \Vert \le 4 H_* \beta_P ^2$$
thanks to Lemma \ref{lem:projection-kernel}. Combining verifies 
$$\Vert F_P -1 \Vert + \Vert \Pi_P - J_P J_P ^* \Vert \le 12 H_* \beta_P ^2 .$$
To verify \eqref{eq:projection-energy}, we recall the bound \eqref{eq:projection-energy-norm}, which shows the upper bound. The lower bound follows by the reverse triangle inequality. Indeed,
$$\left|\Vert M_h ^{1/2}F_P  f \Vert - \Vert M_h ^{1/2} f \Vert \right| \le \Vert M_h ^{1/2} (F_P -1) f \Vert \le C \beta_P ^2 \Vert M_h ^{1/2}f \Vert,$$
where the second inequality is just Lemma \ref{lem:projection-energy-error}. It remains to verify $\Pi_P \cF \subseteq \cD(H(P))$. Note that for every $\xi \in \cF$, $G_P ^{-1}J_P ^* \xi \in L^2(B_r)$, and so 
$$\Pi_P \xi = J_P (G_P ^{-1} J_P ^*) \xi \in \cD(H(P)),$$
thanks to Proposition \ref{prop:intertwine}. Invariance of $\cQ$ under $\Pi_P$ and its complement follows.
\end{proof}

\section{Energy estimates on $\sC_P$ and cross-terms}
\label{sec:channel_energy_estimates}
We use the notation of Section~\ref{sec:projection_construction} and put $A_P=H(P)-E_*$. We recall the notation  $E_*=E(0)+1$ and
\[
 S_Pf=\int_{B_r}\de q\,f(q)g(P-q)\psi_q.
\]
The goal of this section is to prove the following.

\begin{proposition}\label{prop:channel-estimates}
There are constants $C<\infty$ and $P_0<\infty$ such that, for every
$P\in\bbR^3$ with $|P|\ge P_0$ and every $f\in L^2(B_r)$,
\begin{equation}
     \label{eq:channel-lower-bound}
     A_P[U_Pf]\ge(1-C\beta_P^2)\ip{f}{M_hf}
\end{equation}
\begin{equation}
    \label{eq:channel-cross-bound}
    \norm{Q_PA_PU_Pf}\le C\beta_P\ip{f}{M_hf}^{1/2}.
\end{equation}
In particular, for every $y\in\sC_P^\perp\cap\cQ$,
\begin{equation}\label{eq:channel_bound_sum}
 \begin{split}
 A_P[U_Pf+y]
 &\ge(1-C\beta_P^2)\ip{f}{M_hf}+A_P[y]
 -2C\beta_P\norm{y}\ip{f}{M_hf}^{1/2}.
 \end{split}
\end{equation}
The constants are independent of $P$, $f$ and $y$.
\end{proposition}

\begin{remark}
    The comparison in \eqref{eq:projection-energy} is used to express the preceding estimates in terms of $\langle f,M_hf\rangle$. The proof below first gives these estimates with $\langle F_Pf,M_hF_Pf\rangle$ instead. This is already sufficient for the argument in Section \ref{sec:finale}. Indeed, the lower bound on $\sC_P$ and the cross-term estimate only need to involve the same nonnegative quantity so that the cross-term error can be absorbed. We pass back to $f$ for aesthetic reasons and because it does not come with too much additional work.
\end{remark}

\begin{remark}
We will see that the key input in the preceding estimates is the finiteness of $\beta_P^2$. Indeed, it allows us to use the weighted Cauchy-Schwarz inequality
\[
    \int_{B_r}\mathrm dq \, |f(q)|b_P(q)
    \le \beta_P\langle f,M_hf\rangle^{1/2},
\]
thereby controlling the error terms relative to the low-energy action, which may be arbitrarily small. Since $h(q)$ is comparable to $|q|^2$ near zero and $b_P(q)\to b_P(0)>0$ as $q\to0$, the integrand has a singularity of order $|q|^{-2}$. This singularity is locally integrable precisely when $d>2$. Thus, in dimensions one and two, $\beta_P$ would be infinite and the present estimate would fail. In dimension three, we moreover have $\beta_P=O(|P|^{-1})$, which makes the resulting errors small enough to absorb for sufficiently large $|P|$.
\end{remark}

To get started, we first estimate the error term $S_P$ from the identity in Section~\ref{sec:one_phonon_channel}.

\begin{lemma}\label{lem:channel-S-estimates}
For every $|P|>r$ and every $f\in L^2(B_r)$,
\begin{equation}\nonumber
 \norm{S_Pf}\le\beta_P\ip{f}{M_hf}^{1/2}, \qquad \big|\ip{f}{J_P^*S_Pf}\big|
 \le2\beta_P^2\ip{f}{M_hf}.
\end{equation}
\end{lemma}
\begin{proof}
    By the triangle inequality for Bochner integrals we have
    $$\Vert S_P f \Vert \le \int_{B_r} \de q \, |f(q)g(P-q)|=\int_{B_r} \de q \, |f(q)|\sqrt{h(q)}\frac{|g(P-q)|}{\sqrt{h(q)}}.$$
    An application of Cauchy-Schwarz yields the first bound claimed.

    The second bound is obtained by the explicit formulas for $S_P$ and $J_P$. Indeed, for $\xi \in D(\cN_{B_r (P)} ^{1/2})$ and $u \in L^2(B_r)$ we obtain as before
    $$\langle J_P u, \xi \rangle = \int_{B_r} \de q \, \overline{u(q)} \langle \psi_q, a(P-q)\xi \rangle.$$
    We can in particular choose $\xi = \psi_q$ for any $q \in B_r$. Indeed, by pull-through,
    $$\Vert a(P-q') \psi_q \Vert \le 2 |g(P-q')|,$$
    which shows $\psi_q \in D(\cN_{B_r (P)} ^{1/2})$ thanks to Proposition \ref{lem:AK}. The claim is now a direct calculation. Write
    \begin{equation}
        \nonumber
        \begin{split}
            \langle f, J_P ^* S_P f\rangle = \langle J_P f, S_P f\rangle &=  \int_{B_r} \de q f(q)g(P-q) \langle J_Pf,\psi_q \rangle  \\
            &= \int_{B_r} \de q \int_{B_r} \de q' f(q)g(P-q) \overline{f(q')} \langle \psi_{q'},a(P-q')\psi_q \rangle.
        \end{split}
    \end{equation}
    Using Cauchy-Schwarz on the inner product and the pull-through bound $\Vert a(P-q')\psi_{q}\Vert \le 2|b_P(q')|$ yields the overall bound
    $$|\langle f, J_P ^* S_P f\rangle| \le 2\left( \int_{B_r} \de q \,  |f(q)| |b_P (q)| \right)^2 \le 2 \beta_P ^2 \langle f,M_h f\rangle$$
    by multiplying and dividing in the middle integral by $\sqrt{h(q)}$.
\end{proof}

The next lemma gives the energy estimate before replacing $J_P$ with $U_P$.
\begin{lemma}\label{lem:channel-before-normalization}
For all sufficiently large $|P|$ and every $f\in L^2(B_r)$,
\begin{align}
 A_P[J_Pf]
 &=\ip{f}{M_hf}
   +\operatorname{Re}\ip{f}{R_PM_hf}
   +\operatorname{Re}\ip{f}{J_P^*S_Pf},
 \label{eq:channel_action}\\
 Q_PA_PJ_Pf&=Q_PS_Pf.
 \nonumber
\end{align}
In particular, there is a constant $C<\infty$, independent of $P$ and $f$, such that
\begin{equation}\label{eq:distance_diagonal_A_p_action}
 \big|A_P[J_Pf]-\ip{f}{M_hf}\big|
 \le C\beta_P^2\ip{f}{M_hf}.
\end{equation}
\end{lemma}
\begin{proof}
    \eqref{eq:channel_action} follows from a direct application of Proposition \ref{prop:intertwine}. Indeed, applying \eqref{eq:intertwine} to $f$, taking the inner product with $J_P f$ and recalling Lemma \ref{lem:projection-kernel} yields
    $$A_P[J_Pf]= \langle J_Pf, J_P M_h f + S_Pf \rangle = \ip{f}{M_hf}
   +\operatorname{Re}\ip{f}{R_PM_hf}
   +\operatorname{Re}\ip{f}{J_P^*S_Pf},$$
   verifying the first claimed equality.
   To show the second equality, we again use \eqref{eq:intertwine} to see
   $$A_P J_P f = J_P M_h f + S_Pf.$$
   Because $Q_P$ projects onto $\sC^\perp _P$, the first term becomes zero. 
   Take $|P|$ large enough so that Proposition \ref{prop:projection} applies. To show \eqref{eq:distance_diagonal_A_p_action}, we need to estimate the last two terms in \eqref{eq:channel_action}. For the first one we use Lemma \ref{lem:projection-energy-error}. For the second term, use Lemma \ref{lem:channel-S-estimates}.
\end{proof}

\begin{proof}[Proof of Proposition \ref{prop:channel-estimates}]
Fix $f \in L^2(B_r)$. Take  $|P|$ large enough so that the conclusions from Section \ref{sec:projection_construction} apply. Proposition \ref{prop:projection} allows us to find a constant $C_1 >0 $ such that
$$(1-C_1 \beta_P ^2 )\langle f, M_h f \rangle \le \langle F_P f, M_h F_P f \rangle \le (1+C_1 \beta_P^2)\langle f,M_h f \rangle.$$
Next, use Lemma \ref{lem:channel-before-normalization} to find $C_2 > 0$ such that 
$$A_P [U_P f] \ge (1-C_2 \beta_P ^2) \langle F_P f, M_h F_P f \rangle.$$
The last bound holds because $U_Pf = J_P F_P f$.
Combining the last two displays yields
$$A_P [U_P f] \ge (1-C_1 \beta_P ^2)(1-C_2 \beta_P ^2) \langle f, M_h f \rangle.$$
This verifies \eqref{eq:channel-lower-bound} with $C=C_1+C_2$ because $(1-C_1)(1-C_2) \ge (1-(C_1+C_2))$ for $C_1,C_2 \ge 0$.
To show \eqref{eq:channel-cross-bound}, simply observe that $Q_P$ is an orthogonal projection, meaning in particular that $\Vert Q_P \xi \Vert \le \Vert \xi \Vert$ for any $\xi$. Therefore, the claim follows from Lemma \ref{lem:channel-S-estimates} as we will see now. Fix any $y \in \sC^\perp _P \cap \cQ$. Then, by what was just shown
\begin{equation}
    \label{eq:explicit_cross_term_on_channel}
    \begin{split}
        A_P[U_Pf+y] \ge (1- C \beta_P ^2) \langle f,M_h f\rangle + A_P[y] + 2 \mathrm{Re}\langle A_P U_Pf, y\rangle.
    \end{split}
\end{equation}
Because $Q_P y= y$, $\langle A_P U_Pf, y\rangle = \langle Q_P A_P U_Pf, y\rangle$. But now $Q_P A_P U_P f = Q_P S_P F_P f$, which yields the upper bound
$$|\langle Q_P A_P U_Pf, y\rangle|\le \Vert y \Vert \,  \Vert S_P F_Pf \Vert \le C\Vert y \Vert \beta_P \langle f,M_hf \rangle^{1/2}$$
by Lemma \ref{lem:channel-S-estimates}.
Combining with \eqref{eq:explicit_cross_term_on_channel} yields finally
$$ A_P[U_Pf+y] \ge (1- C \beta_P ^2) \langle f,M_h f\rangle + A_P[y] -2 C\Vert y \Vert \beta_P \langle f,M_hf \rangle^{1/2},$$
exactly as claimed.
   \end{proof}

\section{Spectral gap of $H(P)$ on $\sC_P^\perp$}
\label{sec:orthogonal_to_channel}

Recall that $A_P=H(P)-E(0)-1$ and $Q_P$ is the orthogonal projection onto $\sC_P^\perp$. This section verifies that $H(P)$ is coercive with constant at least $E(0)+1+\eta$ for some $\eta >0$ on $\sC_P ^\perp$.

\begin{proposition}\label{prop:complement-gap}
There are constants $\eta>0$ and $P_0<\infty$ such that, for every $P\in\bbR^3$ with $|P|\ge P_0$ and every $\xi\in\sC_P^\perp\cap\cQ$,
\begin{equation}\label{eq:physical-complement-gap}
 A_P[\xi]\ge\eta\norm{\xi}^2.
\end{equation}
The constants may depend on $\alpha$, but not on $P$ or $\xi$.
\end{proposition}

We first decompose the orthogonal complement of the corresponding cutoff space into four subspaces and estimate the energy on each of them. We then compare the cutoff space with $\sC_P$.

\subsection{Decomposing the cutoff Hamiltonian}

Fix $\Lambda>0$. We split the Fock space into high and low momenta. To be more specific, write
\[
 \cF=\cF_<\otimes\cF_>,
 \qquad
 \cF_<=\Gamma_{\mathrm s}(L^2(B_\Lambda)),
 \qquad
 \cF_>=\Gamma_{\mathrm s}(L^2(B_\Lambda^c)).
\]
The phonons with momentum outside $B_\Lambda$ do not interact in $H_\Lambda(P)$. Let $\cN_<$ and $P_{\mathrm f,<}$ denote the number and field momentum operators on $\cF_<$. On this space, put
\begin{equation}\label{eq:inner-cutoff-H}
 h_\Lambda(q)=\frac12(q-P_{\mathrm f,<})^2+\cN_<+\Phi_<(g_\Lambda),
 \qquad e_\Lambda(q)=\inf\operatorname{spec}h_\Lambda(q).
\end{equation}
Here $\Phi_<(g_\Lambda)$ is the interaction acting on $\cF_<$. For $n\ge0$, write
\[
 \cF_\Lambda^{(n)}
 =\cF_<\otimes L^2_{\mathrm{sym}}((B_\Lambda^c)^n),
 \qquad
 \cF=\bigoplus_{n=0}^{\infty}\cF_\Lambda^{(n)}.
\]
Thus $\cF_\Lambda^{(n)}$ consists of states with exactly $n$ phonons outside $B_\Lambda$ and any number inside $B_\Lambda$. For $n=0$, the second factor is the vacuum space $\bbC\Omega_>$.

\begin{lemma}\label{lem:cutoff-number-sectors}
Each $\cF_\Lambda^{(n)}$ is invariant under $H_\Lambda(P)$. 
For $\xi \in \cD(H_\Lambda(P))\cap\cF_\Lambda ^{(n)} \simeq L^2 _{\mathrm{sym}}( (B_\Lambda^c)^n;\cF_<)$, the
restriction of $H_\Lambda(P)$ to $\cF_\Lambda ^{(n)}$ is given by
\begin{equation}\label{eq:cutoff-number-sectors}
 \bigl(H_\Lambda(P)\xi\bigr)(k_1,\ldots,k_n)
 =\left[h_\Lambda\left(P-\sum_{j=1}^n k_j\right)+n\right]
   \xi(k_1,\ldots,k_n).
\end{equation}
For $n=0$, this restriction is $h_\Lambda(P)$. The associated orthogonal projections preserve $\cQ$, and the cutoff quadratic form has no cross terms between different values of $n$.
\end{lemma}

\begin{proof}
    As with Fock space we split $\cN = \cN_> + \cN_<$ as well as $P_f = P_{f,>}+P_{f,<}$. Here, we define $\cN_>$ as the number operator restricted to $\cF_>$ and similarly for $P_{f,>}$. Because $g_\Lambda$ vanishes outside of $B_\Lambda$, $\Phi(g_\Lambda)= \Phi_<(g_\Lambda)\otimes 1$.
    The cutoff Hamiltonian then reads
    $$H_\Lambda(P) = \frac 12 (P-P_{f,<}-P_{f,>})^2 + \cN_> + \cN_< + \Phi_<(g_\Lambda).$$
    On $\cF^{(n)}_\Lambda$, $P_{f,>}$ and $\cN_>$ act as multiplication by $\sum_{j=1}^n k_j$ and $n$, respectively. This verifies \eqref{eq:cutoff-number-sectors} by first considering finite-particle sectors and then extending to the full domain by closure of $H_\Lambda (P)$. Moreover, note that $\cF^{(n)}_\Lambda$ is reducing. Therefore, orthogonal projections preserve $\cQ$ and $H_\Lambda(P)$ commutes with the projection onto $\cF_\Lambda ^{(n)}$. This also implies that the quadratic form has no cross terms for $n \neq m$, which finishes the proof.
\end{proof}

Let $m_r=\inf_{|q|\ge r}(E(q)-E(0))>0$, and recall that $\gamma>0$ is the local spectral gap fixed in Section~\ref{sec:prelims}. We now transfer ground state (energy) properties of $H_\Lambda (P)$ to $h_\Lambda (P)$.

\begin{lemma}\label{lem:inner-cutoff-estimates}
There is $\Lambda_0<\infty$ such that, for every $\Lambda\ge\Lambda_0$, 
\begin{itemize}
    \item $e_\Lambda(0)=E_\Lambda(0),$
    \item $e_\Lambda(q)\ge e_\Lambda(0)$ for $q\in\bbR^3$,
    \item $e_\Lambda(q)-e_\Lambda(0)\ge \frac{m_r}{2}$ for $|q|\ge r$.
\end{itemize}
For $|q|\le2r$, one has $e_\Lambda(q)=E_\Lambda(q)$, $\psi_{\Lambda,q}=\psi^<_{\Lambda,q}\otimes\Omega_>$ as well as
\begin{equation}\label{eq:inner-cutoff-ground-state}
 h_\Lambda(q)-e_\Lambda(q)
 \ge\frac\gamma2
 \bigl(1-|\psi^<_{\Lambda,q}\rangle\langle\psi^<_{\Lambda,q}|\bigr),
\end{equation}
where $\psi^<_{\Lambda,q}$ is the normalized ground state of $h_\Lambda(q)$. Finally, there is $C_{\alpha,\Lambda}<\infty$ such that
\begin{equation}\nonumber
 d_\Lambda(P)=e_\Lambda(P)-e_\Lambda(0)-1
 \ge\min\left\{\frac{|P|^2}{8},\frac{|P|}{4\Lambda}\right\}
       -C_{\alpha,\Lambda}.
\end{equation}
In particular, $d_\Lambda(P)\to\infty$ as $|P|\to\infty$ for each fixed $\Lambda$.
\end{lemma}

\begin{proof}
    First, note that for any $q \in \bbR^3$
    $$e_\Lambda (q) \ge E_\Lambda(q) \ge E_\Lambda (0).$$
    The first inequality holds because $h_\Lambda (q)$ can be thought of as the restriction of $H_\Lambda (q)$ to $\cF_<$. The second inequality is \ref{I:gross}. We now show $e_\Lambda (0) = E_\Lambda (0)$. By \eqref{eq:cutoff-number-sectors} and the last displayed equation we find that $H_\Lambda(0)$ has energy at least $E_\Lambda (0) +1$ on $\bigoplus_{n \ge 1} \cF^{(n)}_\Lambda$. But also $H_\Lambda(0) \ge E_\Lambda (0)$ by definition, which implies  $\psi_{\Lambda,0} \in \cF_<$. But since $H_\Lambda (0) = h_\Lambda(0)$ on $\cF_<$, $e_\Lambda(0) \le E_\Lambda (0)$, showing the desired equality. This directly implies for $|q|\ge r$ and $\Lambda$ large enough
    $$e_\Lambda(q)-e_\Lambda (0) \ge  E_\Lambda(q) - E_\Lambda (0) \ge \frac{m_r}{2},$$ 
    exactly as claimed. 
    We continue with showing $e_\Lambda (q) = E_\Lambda (q)$ for $|q|\le 2r$, where it is sufficient to verify $\psi_{\Lambda,q} \in \cF_<$. This holds by a similar argument as before. Indeed, $H_\Lambda (P) \ge E_\Lambda (0) +1$ if at least one high-momentum phonon is present. But for $|q|\le 2r$ and $\Lambda$ sufficiently large, $E_\Lambda (q) < E_\Lambda(0)+1$, which means $\psi_{\Lambda,q}$ must be contained entirely in $\cF_<$. Therefore,
    $\psi_{\Lambda,q}= \psi_{\Lambda,q}^< \otimes \Omega_>$
    and $e_\Lambda (q) = E_\Lambda (q)$ for $|q|\le 2r$. \eqref{eq:inner-cutoff-ground-state} is then a simple consequence of this, as
    $$h_\Lambda(q) - e_\Lambda(q) = H_\Lambda(q) - E_\Lambda (q),$$
    where we interpret the right-hand side operator  restricted to $\cF^{(0)}_\Lambda$. Then, using Proposition \ref{lem:riesz-local} shows the claim by our choices of $r$ and $\gamma$ at the beginning of Section \ref{sec:prelims}. Finally,
    $$h_\Lambda (P) \ge \frac 1 2 (P-P_f)^2 + \frac 1 2 \cN_< - 2 \Vert g_\Lambda \Vert^2,$$
    as operators on $\cF(L^2(B_\Lambda))$, which follows by the standard bound on Segal fields for regular interactions. We now bound the operator on all $n$ particle sectors separately. We also recall that all momenta are in $B_\Lambda$, since $h_\Lambda (q)$ acts on $\cF_<$, meaning $|P_f|\le n\Lambda$. If $n < |P|/2\Lambda$, then $|P-P_f|\ge |P|/2$, and so the kinetic term is at least $|P|^2 /8$. On the other hand, if $n \ge |P|/2 \Lambda$, $\cN_< \ge |P|/2\Lambda$. Thus, the desired lower bound for $d_\Lambda (P)$ follows by subtracting $e_\Lambda(0)+1$ on both sides and  choosing $C_{\alpha,\Lambda}= 2\Vert g_\Lambda \Vert^2+ 1+e_\Lambda(0)$.
\end{proof}

In the sector with one phonon outside $B_\Lambda$, use the variable $q=P-k$ to identify
\[
 \cF_\Lambda^{(1)}\simeq L^2(B_\Lambda (P)^c;\cF_<).
\]
In other words, a function $f \in L^2( B_\Lambda ^c;\cF_<)$ corresponds to $L^2(B_\Lambda (P)^c;\cF_<) \ni\Tilde f (q) = f(P-q)$. This identification is just for simplified notation, because now $q$ denotes the remaining momentum not carried by the hard phonon and hence makes the action of $H_\Lambda (P)$ diagonal. Indeed, whereas we had before (with $\xi \in L^2(B_\Lambda ^c;\cF_<)$)
$$H_\Lambda(P) \xi (k) = (h_\Lambda (P-k) +1) \xi(k),$$
we have on $L^2(B_\Lambda (P)^c;\cF_<) \ni \Tilde \xi$ that $H_\Lambda (P) \Tilde \xi (q) = (h_\Lambda(q)+1) \Tilde \xi(q)$. This change of coordinates does not change the spectrum, as the mapping is unitary.
Suppose now that $|P|> \Lambda +r$. This implies  $B_r\subset B_\Lambda (P)^c$, and so we can define $J^0_{\Lambda,P}: L^2(B_r) \to L^2(B_\Lambda (P)^c;\cF_<)$ by
\begin{equation}\nonumber
 (J^0_{\Lambda,P}f)(q)
 =\begin{cases}
    f(q)\psi^<_{\Lambda,q},&q\in B_r,\\
    0,&q\in B_\Lambda (P)^c\setminus B_r,
   \end{cases}
\end{equation}
and extend by zero components in all other sectors to $\cF$. This is an isometry, which can be verified by a direct computation. Write
\[
 \sC^0_{\Lambda,P}=\operatorname{Ran}J^0_{\Lambda,P},
 \qquad
 \Pi^0_{\Lambda,P}=J^0_{\Lambda,P}(J^0_{\Lambda,P})^*,
 \qquad Q^0_{\Lambda,P}=1-\Pi^0_{\Lambda,P},
\]
and set $A_{\Lambda,P}=H_\Lambda(P)-e_\Lambda(0)-1$.
Suppressing $\Lambda$ and $P$ in their notation, define the four subspaces
\begin{align*}
 \mathcal{V}_0&=\cF_\Lambda^{(0)},\\
 \mathcal{V}_{\mathrm{out}}
 &=L^2(B_\Lambda (P)^c\setminus B_r;\cF_<),\\
 \mathcal{V}_{\mathrm{exc}}
 &=\bigl\{u\in L^2(B_r;\cF_<):
       \ip{\psi^<_{\Lambda,q}}{u(q)}=0
       \text{ for almost every }q\bigr\},\\
 \mathcal{V}_{\ge2}&=\bigoplus_{n\ge2}\cF_\Lambda^{(n)}.
\end{align*}
The two  spaces $\mathcal V_\mathrm{out}, \mathcal V_\mathrm{exc}$ are regarded as subspaces of
$\cF_\Lambda^{(1)}$ by extension by zero. They describe, respectively, residual momentum outside $B_r$ and an inner state orthogonal to the ground state at residual momentum inside $B_r$. $\mathcal V_0$ describes the spaces where no hard phonon is present, and $\mathcal V_{\ge 2}$ is the subspace where more than one hard phonon is present. Intuitively, the first space needs to have too many soft phonons to carry large momentum, whereas two or more hard phonons have too much of a cost.

\begin{proposition}
\label{prop:cutoff-complement-decomposition}
For every sufficiently large $\Lambda$ and every $|P|>\Lambda+r$,
\begin{equation}\label{eq:cutoff-complement-direct-sum}
 (\sC^0_{\Lambda,P})^\perp
 =\mathcal{V}_0\oplus\mathcal{V}_{\mathrm{out}}
       \oplus\mathcal{V}_{\mathrm{exc}}\oplus\mathcal{V}_{\ge2}.
\end{equation}
The sum is orthogonal, and all subspaces reduce $H_\Lambda (P)$. The projections onto these four subspaces preserve $\cQ$.
The individual energy bounds are
\begin{enumerate}
    \item $A_{\Lambda,P}[x_0]\ge d_\Lambda(P)\norm{x_0}^2$ with  $x_0 \in \mathcal V_0 \cap \cQ$,
    \item $A_{\Lambda,P}[x_{\mathrm{out}}] \ge\frac{m_r}{2}\norm{x_{\mathrm{out}}}^2$ with  $x_\mathrm{out} \in \mathcal V_\mathrm{out}\cap \cQ$,
    \item $ A_{\Lambda,P}[x_{\mathrm{exc}}]\ge\frac\gamma2\norm{x_{\mathrm{exc}}}^2$ with  $x_\mathrm{exc} \in \mathcal V_\mathrm{exc}\cap \cQ$,
    \item $A_{\Lambda,P}[x_{\ge2}]\ge\norm{x_{\ge2}}^2$ with  $x_{\ge2} \in \mathcal V_{\ge2}\cap \cQ$.
\end{enumerate}

Consequently, with $\eta_0=\frac14\min\{m_r,\gamma,1\}>0$, for every sufficiently large $\Lambda$ there is $P_0(\Lambda)>\Lambda+r$ such that
\begin{equation}\label{eq:cutoff-complement-gap}
 A_{\Lambda,P}[y]\ge\eta_0\norm{y}^2,
 \qquad
 \text{where } y\in(\sC^0_{\Lambda,P})^\perp\cap\cQ,\quad
       |P|\ge P_0(\Lambda).
\end{equation}
The gap $\eta_0$ is independent of $\Lambda$.
\end{proposition}

\begin{proof}
    In order to prove \eqref{eq:cutoff-complement-direct-sum} it remains to verify the orthogonal decomposition $L^2(B_r ; \cF_<) = \mathcal V _\mathrm{exc} \oplus \sC^0 _{\Lambda,P}$. By what we have shown already, the projection onto $\sC^0 _{\Lambda,P}$ is $\Pi_{\Lambda,P}^0$. Therefore, for any $\xi \in L^2(B_r;\cF_<)$,
    $$\Pi_{\Lambda,P}^0 \xi (q) = \psi_{\Lambda,q}\langle \psi_{\Lambda,q},\xi(q)  \rangle.$$
    The last display is the zero vector if, and only if, $\langle \xi(q), \psi_{\Lambda,q} \rangle = 0$ for almost every $q$. This is exactly $\mathcal V_{\mathrm{exc}}$, which verifies \eqref{eq:cutoff-complement-direct-sum}.
    All subspaces reduce $H_\Lambda (P)$. Indeed, $\Phi(g_\Lambda)$ can not create or destroy a phonon with momentum $\ge \Lambda$, which means that outer phonon numbers are conserved. This means that $H_\Lambda(P)$ leaves outer phonons untouched, and so the spaces reduce $H_\Lambda (P)$. We next verify the claimed energy estimates.
    First, let $x_0 \in \mathcal V _0 \cap \cQ$. Then
    $$A_{\Lambda,P}[x_0] \ge (e_\Lambda(P)-e_\Lambda (0) -1) [x_0] = d_\Lambda(P)\Vert x_0 \Vert^2,$$
    as $H_\Lambda (P)$ restricted to $\cF_<$ is just $h_\Lambda (P)$. Next, take $x_\mathrm{out} \in \mathcal V_\mathrm{out}\cap \cQ$. Then, by \eqref{eq:cutoff-number-sectors} we have the formula
    $$A_{\Lambda,P}[x_\mathrm{out}] \ge \int_{B_\Lambda (P)^c \backslash B_r} \de q \,  \langle x_\mathrm{out}(q), \big( e_\Lambda(q) -e_\Lambda(0) \big) x_\mathrm{out}(q)  \rangle \ge \frac{m_r }{2}\Vert  x_\mathrm{out}\Vert^2,$$
    by definition of $m_r$ and the domain we integrate over. This verifies the second claimed energy estimate. Third, pick $x_\mathrm{exc} \in \mathcal V_\mathrm{exc}\cap \cQ$. Split 
    $$x_\mathrm{exc} (q) = \psi_{\Lambda,q}^< \langle  \psi_{\Lambda,q}^<,x_\mathrm{exc} (q) \rangle + \phi_q$$
    with $\phi_q \in (\mathrm{span }{\{\psi_{\Lambda,q}^<} \})^\perp$. By assumption the first term is $0$. Using again the explicit action of $H_\Lambda (P)$ on the soft-phonon Fock space via \eqref{eq:cutoff-number-sectors} yields
    $$A_{\Lambda,P}[x_\mathrm{exc}] = \int_{B_r} \de q \big( h_\Lambda(q) - e_\Lambda(0) \big)[ \phi_q] .$$
    Because $\phi_q$ is orthogonal to the ground state, we have the point-wise bound
    $$\big( h_\Lambda(q) - e_\Lambda(0) \big)[ \phi_q]  \ge \frac\gamma 2 \Vert \phi_q \Vert ^2.$$ Integrating against $q$ on $B_r$ verifies the third energy estimate. Finally, let $x_{\ge 2} \in \mathcal V_{\ge 2}\cap \cQ$. Using once again \eqref{eq:cutoff-number-sectors} shows that $A_{\Lambda, P}$ is lower bounded by $1$ in case at least two hard phonons are present. By definition $\eta_0$ is independent of $\Lambda$, finishing the proof.
\end{proof}

\subsection{Comparison with the full Hamiltonian}

The decomposition \eqref{eq:cutoff-complement-direct-sum} gives an
additive energy estimate for $H_\Lambda(P)$. For $H(P)$ the four subspaces can interact. We use the form comparison from Appendix~\ref{app:basic_froehlich} and the following estimates to obtain a bound on $\sC_P^\perp$.

\begin{lemma}\label{lem:cutoff-projection-comparison}
There is $C<\infty$ such that, for all sufficiently large $\Lambda$ and every $|P|>\Lambda+r$,
\begin{align}
 \norm{J_P-J^0_{\Lambda,P}}^2
 &\le\delta_\Lambda^2+4\norm{b_P}_{L^2(B_r)}^2,
 \label{eq:cutoff-channel-comparison}\\
 \norm{\Pi_P-\Pi^0_{\Lambda,P}}
 &\le C\bigl(\delta_\Lambda+|P|^{-1}\bigr).
 \label{eq:cutoff-projection-comparison}
\end{align}
The constant is independent of $\Lambda$ and $P$.
\end{lemma}

\begin{proof}
To prove \eqref{eq:cutoff-channel-comparison}, we apply Lemma \ref{lem:creation}. Indeed, recall that for $|q|\le r$, $\psi_{\Lambda,q}$ is supported on $\cF_<$. Therefore, $J^0 _{\Lambda,P} = C^{\psi_{\Lambda}}$, with $\psi_\Lambda = \{\psi_{\Lambda,q}: q \in B_r  \}$. An application of \eqref{eq:creationdiff} yields
$$\norm{(J_P-J^0_{\Lambda,P})f }^2 \le \left(
 \sup_{|q|\le r}\norm{\psi_{\Lambda,q}-\psi_q}^2
 +\sup_{|q|\le r}\int_{B_r}\dd q'\,\norm{a(P-q')(\psi_{\Lambda,q}-\psi_q)}^2
 \right)\norm f ^2.$$
 The first term on the right-hand side is bounded by $\delta_\Lambda^2$ by definition (recall Proposition \ref{lem:riesz-local}). For the second term, note that  $a(P-q')\psi_{\Lambda,q} = 0$ for $(q,q') \in B_r \times B_r$, and so $\Vert a(P-q')(\psi_{\Lambda,q}-\psi_q)\Vert = \Vert a(P-q') \psi_q\Vert \le 2b_P (q')$. Therefore 
 $$\norm{(J_P-J^0_{\Lambda,P})f }^2 \le \left(\delta_\Lambda ^2 + 4\Vert b_P \Vert _{L^2 (B_r)}^2 \right) \Vert f \Vert^2,$$
 which is \eqref{eq:cutoff-channel-comparison}. In order to show the second claimed inequality, we choose $|P|$ sufficiently large so that we can apply Proposition \ref{prop:projection}. We then have $\Vert J_P \Vert \le
 \sqrt{3/2}$
 , $\Vert F_P -1 \Vert \le C|P|^{-2}$ as well as $\Vert b_P \Vert_{L^2(B_r)}\le C/|P|$. Therefore,
 $$\Vert U_P - J^0 _{\Lambda,P} \Vert \le \Vert J_P (F_P -1)  \Vert + \Vert J_P - J^0 _{\Lambda,P}\Vert \le C(\delta_\Lambda + |P|^{-1}).$$
 Using that $\Vert UU^* - VV^*\Vert \le 2 \Vert U-V\Vert$ for two isometries $U,V$ on the same initial space, \eqref{eq:cutoff-projection-comparison} is shown (pick $U= J^0 _{\Lambda,P}$, $V= U_P$). Finally, note that we chose all of our constants independent of $\Lambda$ or $P$.
\end{proof}

\begin{lemma}\label{lem:physical-energy-cutoff-spaces}
There are $\Lambda_1<\infty$ and $M<\infty$ such that, for every $\Lambda\ge\Lambda_1$ and every $|P|\ge P_0(\Lambda)$, the cutoff energy on $\sC^0_{\Lambda,P}$ is given by
\begin{equation}\label{eq:cutoff-channel-energy}
 A_{\Lambda,P}[J^0_{\Lambda,P}f]
 =\int_{B_r}\de q\,\bigl(e_\Lambda(q)-e_\Lambda(0)\bigr)|f(q)|^2
 \qquad(f\in L^2(B_r)).
\end{equation}
For the full Hamiltonian, one has 
$$A_P[y]\ge\eta_1\norm{y}^2, \qquad \qquad \qquad (A_P+2)[z]\le M^2\norm{z}^2$$
for $y\in (\sC^0_{\Lambda,P})^\perp\cap\cQ$ and $z\in\sC^0_{\Lambda,P}$,  where $\eta_1=\eta_0/2$. Moreover, $\sC^0_{\Lambda,P}\subset\cQ$, and $\Pi^0_{\Lambda,P}$ and $Q^0_{\Lambda,P}$ preserve $\cQ$. The constants $\eta_1$ and $M$ are independent of $\Lambda$ and $P$.
\end{lemma}

\begin{proof}
    Take $f \in L^2(B_r)$. Because $\psi_{\Lambda,q} = \psi_{\Lambda,q}^<\otimes \Omega_>$ and by \eqref{eq:cutoff-number-sectors}, it holds that
    $$H_\Lambda (P) J^0 _{\Lambda,P}f = J^0 _{\Lambda,P} M_{e_\Lambda (\cdot)+1} f.$$
In particular, $\sC^{0}_{\Lambda,P} \subseteq D(H_\Lambda (P))\subseteq \cQ$, and this verifies \eqref{eq:cutoff-channel-energy} by taking the scalar product with $J^0 _{\Lambda,P}f$. By Proposition \ref{prop:cutoff-complement-decomposition} it is known that $\Pi^0 _{\Lambda,P}$ reduces $H_\Lambda (P)$, which means that the projection as well as projecting onto the orthogonal complement preserve $\cQ$. We now apply \eqref{eq:two-sided} to see that for any $y \in (\sC^0 _{\Lambda,P})^\perp$,
$$A_P [y] \ge (1-\epsilon_\Lambda) H_\Lambda (P) [y] - (\epsilon_\Lambda b  + E_*) \Vert y \Vert^2 \ge \left[ e_\Lambda \!(0)\! - \!E (0)+\eta_0\!-\!\epsilon_\Lambda (e_\Lambda (0) +1+\eta_0 + b) \right]\Vert y \Vert^2,$$
where $b$ is chosen as in Proposition \ref{prop:form-comparison}. As $e_\Lambda (0) = E_\Lambda (0) \to E (0)$, the term multiplied by $\epsilon_\Lambda$ can be made arbitrarily small by increasing $\Lambda$. Choose $\Lambda$ big enough so that 
$$\epsilon_\Lambda (e_\Lambda (0) +1+\eta_0 + b) \le \eta_0 /4, \qquad \qquad |E_\Lambda(0) - E(0)| \le \eta_0 /4,$$
which shows $A_P [y] \ge \eta_1 \Vert y \Vert^2$ with $\eta_1 = \eta_0 /2$. It remains to verify  
$$(A_P+2)[z]\le M^2\norm{z}^2$$
for $z \in \sC^0 _{\Lambda,P}$. This follows from another application of  \eqref{eq:two-sided}. Indeed, estimate
$$(A_P+2)[z] = H(P)[z] + (2-E_*) \Vert z \Vert^2 \le (1+\epsilon_\Lambda) \big( H_{\Lambda}(P)[z] + b \Vert z \Vert^2 \big) + (2-E_*-b) \Vert z \Vert^2.$$
Choosing now $z = J^{0}_{\Lambda,P}f$ for some $f \in L^2(B_r)$ yields, using \eqref{eq:cutoff-channel-energy} and that $J^0 _{\Lambda,P}$ is isometric,
$$ H_{\Lambda}(P)[z] \le (C_r + 1+b) \Vert z \Vert^2,$$
where the constant $C_r>0$ is uniform in $\Lambda$. Combining yields the claim.
\end{proof}

We finally compare the two orthogonal complements. Recall that
$A_P\ge-1$, so $A_P+2\ge1$.

\begin{lemma}\label{lem:change-complement}
Assume the conclusions of Lemma~\ref{lem:physical-energy-cutoff-spaces}
and put $d=\norm{\Pi_P-\Pi^0_{\Lambda,P}}<1$.
For every $\xi\in\sC_P^\perp\cap\cQ$
$\norm{\Pi^0_{\Lambda,P}\xi}\le d\norm{\xi}$, and
there is $d_0=d_0(\eta_1,M)>0$ such that
\begin{equation}\nonumber
 A_P[\xi]\ge\frac{\eta_1}{2}\norm{\xi}^2
 \qquad(\xi\in\sC_P^\perp\cap\cQ)
\end{equation}
for all $d\le d_0$. The number $d_0$ is independent of $\Lambda$ and $P$.
\end{lemma}

\begin{proof}
    For convenience we abbreviate $T=A_P + 2 \ge 1$. Take $\xi \in \sC^\perp _P \cap \cQ$. By the previous Lemma, $y = Q^0 _{\Lambda,P} \xi$ and $z = \Pi^0_{\Lambda,P}\xi$ belong to $\cQ$. Clearly $\Pi_P \xi = 0$. This yields the claimed bound on $z$, as
    $$z = \Pi^0 _{\Lambda,P} \xi = (\Pi^0 _{\Lambda,P}-\Pi_P) \xi.$$
    Since $y \perp z$ we therefore find the simple (no-mixing) bound 
    $\Vert y \Vert^2 \ge (1-d^2)\Vert \xi\Vert^2.$
    Thanks to Lemma \ref{lem:physical-energy-cutoff-spaces} we moreover find the bounds
    $$\Vert T^{1/2} y\Vert \ge \sqrt{\eta_1+2}\Vert y \Vert, \qquad \Vert T^{1/2}z \Vert \le M \Vert z \Vert.$$
    Using the reverse triangle inequality yields
    $$T[\xi]^{1/2}\ge \Vert T^{1/2}y\Vert - \Vert T^{1/2}z \Vert \ge \left( \sqrt{\eta_1 + 2} \sqrt{1-d^2}-Md \right)\Vert \xi \Vert.$$
    Choose $\Lambda$ and $|P|$ large enough that 
    $$\sqrt{\eta_1 + 2} \sqrt{1-d^2}-Md \ge \sqrt{2+\eta_1/2}.$$
    Note that this choice only depends on the uniform constants $\eta_1$ and $M$. Squaring both sides and after subtracting $2 \Vert \xi \Vert^2$ on both sides yields the claim.
\end{proof}

\begin{proof}[Proof of Proposition \ref{prop:complement-gap}]
    Take $\eta_1,M$ and $d_0$ as in Lemma \ref{lem:physical-energy-cutoff-spaces} and Lemma \ref{lem:change-complement}. Moreover, let $C_*$ be the constant used in \eqref{eq:cutoff-projection-comparison}. All constants just named were shown to be independent of any sufficiently large cutoff $\Lambda$. Because $\delta_\Lambda \to 0$, we first fix $\Lambda \ge \Lambda_1$ large enough so that all previous conclusions apply, and $C_* \delta_\Lambda \le d_0 /2$. Using this cutoff, choose $P_0 \ge P_0 (\Lambda)$ such that Proposition \ref{prop:projection} applies and $C_* |P|^{-1}\le d_0 /2$. Then \eqref{eq:cutoff-projection-comparison} yields
    $$\Vert \Pi_P - \Pi^{0}_{\Lambda,P}\Vert \le d_0,$$
    in which case Lemma \ref{lem:change-complement} yields
    $$A_P[\xi]\ge \frac{\eta_1} 2 \Vert \xi \Vert^2$$
    for any $\xi \in \sC^\perp _P \cap \cQ$ and $|P|\ge P_0$. This yields the claim with $\eta = \eta_1 / 2 = \eta_0 /4 > 0$.
\end{proof}

\section{Putting everything together}
\label{sec:finale}
\begin{proof}[Proof of Theorem \ref{thm:main}]
    Take $\alpha >0$ and recall the abbreviation $A_P = H(P) - E_*$, where $E_* = E(0) +1$. Take $C>0$ and $\eta >0$ as the constants found in Proposition \ref{prop:channel-estimates} and Proposition \ref{prop:complement-gap}. Moreover, $|P|\ge P_0$ as in Proposition \ref{prop:channel-estimates}. Recall also that 
    $$\beta_P ^2 = \int_{B_r} \de q \, \frac{|g(P-q)|^2}{E(q)-E(0)} = \frac{\alpha}{2\pi^2}\int_{B_r} \de q \, \frac{1}{(E(q)-E(0)) |P-q|^2}  \to 0$$
    as $|P|\to \infty$. We can therefore choose $P_0 \le P_\alpha < +\infty$ such that the conclusions of Propositions  \ref{prop:projection}, \ref{prop:channel-estimates} and \ref{prop:complement-gap} hold true for all $|P|\ge P_\alpha$. In particular,
    \begin{equation}
        \label{eq:final_constants_bound}
        \left(C+ \frac{2C^2}{\eta }\right)\beta_P^2\le \frac 1 2 .
    \end{equation}
    Let us now take $\xi \in \cQ$. Thanks to Proposition \ref{prop:projection} we may find a function $f_\xi \in L^2(B_r)$ that recovers $\xi$ when applying $U_P$ to it and a remainder lying in the orthogonal complement of $\sC_P$. More precisely, we have the decomposition
    $$\xi = U_P f_\xi + y_\xi$$
    with $f_\xi = U_P ^* \xi$ and $y_\xi = Q_P \xi \in \sC_P ^\perp$. Most notably, $U_P ^* \xi \in L^2(B_r)$ as well as $y_\xi \in \sC_P ^\perp \cap \cQ$. Proposition \ref{prop:channel-estimates} and Proposition \ref{prop:complement-gap} now let us estimate the energy on $\sC_P$ as well as cross-terms. Indeed, 
    $$
    A_P[\xi] \ge (1-C\beta_P ^2) \langle f_\xi,M_h f_\xi \rangle - 2 C \beta_P \langle f_\xi,M_h f_\xi \rangle^{1/2} \Vert y_\xi \Vert + \eta \Vert y_\xi \Vert^2,
    $$
    which is just \eqref{eq:channel_bound_sum} followed by Proposition \ref{prop:complement-gap} applied to $A_P[y_\xi]$. Young's inequality applied to the second term moreover yields 
    $$2 C \beta_P \langle f_\xi,M_h f_\xi \rangle^{1/2} \Vert y_\xi \Vert \le \frac{2C^2}{\eta}\beta_P ^2 \langle f_\xi,M_hf_\xi \rangle+\frac \eta 2 \Vert y_\xi \Vert^2.$$
    Combining the previous two bounds gives
    \begin{equation}
        \label{eq:final_calc}
        A_P [\xi] \ge \left( 1- \beta_P ^2 \big(C + 2 \frac{C^2}{\eta}\big) \right) \langle f_\xi,M_h f_\xi\rangle + \frac \eta 2 \Vert y_\xi \Vert^2 \ge \frac 1 2 \left( \langle f_\xi,M_h f_\xi \rangle + \eta \Vert y_\xi \Vert^2 \right) \ge 0, 
    \end{equation}
    where the second inequality follows by equation \eqref{eq:final_constants_bound}. By Rayleigh-Ritz, $E(P) \ge E_*$, which implies by \ref{I:hvz} that $E(P) = E(0)+1$ for all $P$ so that $|P|\ge P_\alpha$, and hence \eqref{eq:main_goal} is verified. To show that no ground state can sit at the edge, assume there exists $\xi \in \Ker(H(P) - E(0)-1)$. Then $\xi \in \cQ$ and $A_P [\xi] = 0$. By \eqref{eq:final_calc} this implies $\langle f_\xi,M_h f_\xi \rangle = 0$. But $E(q) > E(0)$ on $B_r$ almost everywhere, which means that $f_\xi=\delta_0$. But $\delta_0 \notin L^2 (B_r)$, which shows the claim.
\end{proof}

\textbf{Acknowledgements:}
We thank Steffen Polzer, Morris Brooks, Volker Betz and Jonas Lampart for several interesting and enlightening discussions on the large momentum problem. In particular, we thank Benjamin Hinrichs for helpful discussions and valuable feedback on the manuscript. The author acknowledges support by the German Research Foundation (DFG) within the scientific network ‘A(E)MP – Appearance of the Effective Mass in Polaron Models’ (Grant No. 569490025). TS also gratefully acknowledges support by DFG grant No 535662048. Finally, we thank the Center of Mathematical Sciences and Applications (Harvard University) for its support and for providing a great research atmosphere.

\textbf{AI disclosure:}
ChatGPT was used to assist with literature searches, improve the presentation and verify correctness. All mathematical arguments in this manuscript were written by the author, who takes full responsibility for their correctness.

\appendix

\section{Basic properties of the Fr\"ohlich Hamiltonian}
\label{app:basic_froehlich}
In this appendix we want to collect basic properties of the Hamiltonians we work with throughout. We find it sometimes hard to get a good overview of the results, most notably because notation differs quite heavily between different papers. We therefore provide all statements we require in our notation and cite the original source; in case small changes are required, we explain briefly how to obtain the results from the literature.

\begin{proposition}\label{lem:LY}
There is a numerical constant $C$ such that, for every $P\in\R^3$, every $0<\Lambda<S<\infty$, and every $\xi\in\cQ$,
\begin{equation}\label{eq:LY}
 |\Phi(g_S-g_\Lambda)[\xi]|
 \le C\sqrt\alpha\,(\Lambda^{-1}-S^{-1})^{1/2}
       \norm{(K(P)+1)^{1/2}\xi}^2.
\end{equation}
The constant is independent of $P,\Lambda,S$ and $\alpha$.
\end{proposition}

\begin{proof}
    This is the fiber version of \cite[Lemma 2.1]{GW}. The similar argument provided applies with the particle momentum replaced by $P-P_f$; only numerical constants change. The cutoff dependence follows from
    $$\int_{\Lambda < |k|<S} \de k \frac{|g(k)|^2}{|k|^2} = \frac{2\alpha}{\pi}(\Lambda^{-1}-S^{-1}),$$
    showing the claim.
\end{proof}

\begin{proposition}\label{prop:form-comparison}
There are constants $c_0>0$, $C_0<\infty$, and $b>C_0+1$, depending only on $\alpha$, such that
\begin{equation}\label{eq:coercive}
 H_\Lambda(P)\ge c_0K(P)-C_0
\end{equation}
for every $P$ and every sufficiently large $\Lambda$. The same estimate holds for $H(P)$. Furthermore, there are numbers $\eps_\Lambda\to0$, with $0\le\eps_\Lambda\le C_\alpha \Lambda^{-1/2}$, such that
\begin{equation}\label{eq:formcompare}
 |(H(P)-H_\Lambda(P))[\xi]|
 \le\eps_\Lambda\bigl(H_\Lambda(P)[\xi]+b\norm\xi^2\bigr)
\end{equation}
for every $\xi\in\cQ$, uniformly in $P$, and hence
\begin{equation}\label{eq:two-sided}
 (1-\eps_\Lambda)(H_\Lambda(P)+b)
 \le H(P)+b
 \le(1+\eps_\Lambda)(H_\Lambda(P)+b).
\end{equation}
In particular, the cutoff fiber Hamiltonians converge in norm-resolvent sense uniformly in the fiber parameter $P$ to the physical fiber Hamiltonian and
\begin{equation}\label{eq:uniform-energy-convergence}
 \sup_{P\in\R^3}|E_\Lambda(P)-E(P)|\longrightarrow0,
\end{equation}
where the convergence rate is $O_\alpha(\Lambda^{-1/2})$.
\end{proposition}

\begin{proof}
    The form construction and the lower bound \eqref{eq:coercive} can be derived by applying \cite[Theorem 2.2 and Corollary 2.3]{GW} in each fiber and using \eqref{eq:LY}. Taking $S\to\infty$ in \eqref{eq:LY} yields
    $$|(H_\Lambda(P) - H(P))[\xi]| \le \frac{C_\alpha}{\Lambda^{1/2}}\Vert (K(P)+1)^{1/2}\xi \Vert^2$$
    for any $\xi \in \cQ$. By \eqref{eq:coercive}, $K(P)+1 \le C_1 (H_\Lambda(P)+b)$ uniformly in $P$ and large $\Lambda$, which shows \eqref{eq:formcompare}. To derive \eqref{eq:two-sided}, use $\eqref{eq:formcompare}$, which says that
    $$-\epsilon_\Lambda(H_\Lambda(P)+b) \le H(P)+b - (H_\Lambda(P)+b) \le \epsilon_\Lambda (H_\Lambda(P)+b).$$
    Re-arranging then yields \eqref{eq:two-sided}. Finally, taking infima in \eqref{eq:two-sided} gives the spectral inequalities
    $$(1-\epsilon_\Lambda)(E_\Lambda(P)+b) \le E(P)+b \le (1+\epsilon_\Lambda)(E_\Lambda(P)+b).$$
    Subtracting $E_\Lambda(P)+b$ in the last display yields $|E_\Lambda(P)-E(P)|\ \le \epsilon_\Lambda (E_\Lambda(P)+b)$. But the first inequality in the last display yields
    $$(1-\epsilon_\Lambda)(E_\Lambda(P)+b) \le E(P)+b \le E(0)+b+1.$$
    Thus $|E_\Lambda(P)-E(P)|\ \le \epsilon_\Lambda (E(0)+b+1)/(1-\epsilon_\Lambda),$ showing the claim.
\end{proof}

\begin{proposition}
    \label{prop:uniform_resolvent_conv}
    Let $K \subset \bbR^3$ be compact. Then, for any $z \in \bigcap_{ P \in K} \rho(H(P))$ we have
    $$ \sup\limits_{P \in K} \Vert ( H_\Lambda(P) - z)^{-1} - (H(P) - z)^{-1} \Vert \to 0.$$
\end{proposition}
\begin{proof}
    First choose $b$ as found in  Proposition \ref{prop:form-comparison}. The case of general $z \in \bigcap_{ P \in K} \rho(H(P))$ follows by the resolvent identity once it is verified that $\sup_{P \in K} \Vert (H(P) - z)^{-1}\Vert < \infty$. In order to see this, recall that the family $\{H(P): P \in K\}$ is norm resolvent continuous, which yields the uniform bound. Next, \eqref{eq:two-sided} implies
    $$ \frac{1}{(1+\epsilon_\Lambda)}  ( H_\Lambda(P)+b )^{-1} \le ( H(P)+b )^{-1} \le \frac{1}{(1-\epsilon_\Lambda)} ( H_\Lambda(P)+b )^{-1}.$$
    Since the spectrum of $H(P) + b$ is uniformly bounded from below by some $a >0$ for all $P \in K$, we conclude that
    $$ -\epsilon_\Lambda I \le ( H(P)+b )^{-1} - ( H_\Lambda(P)+b )^{-1} \le \epsilon_\Lambda I.$$
    Because the error is independent of $P$, uniform convergence follows, which was exactly the claim.
\end{proof}

Define
\[
 \mathfrak C=\bigoplus_{n=0}
                C^\infty_{c,\mathrm{sym}}((\mathbb R^3)^n),
\]
where the zeroth summand is $\mathbb C$. Thus $\mathfrak C$ consists of vectors with finitely many nonzero components, each smooth, compactly supported, and symmetric.

\begin{lemma}\label{lem:finite_cutoff_core}
For every $P\in\mathbb R^3$ and finite $\Lambda>0$, $H_\Lambda(P)$ is self-adjoint on $\mathcal D(K(P))$. The graph norms of $H_\Lambda(P)$ and $K(P)$ are equivalent, and $\mathfrak C$ is an operator core for both operators.
\end{lemma}

\begin{proof}
    As $g_\Lambda$ is an $L^2$ function, $\Phi(g_\Lambda)$ is infinitesimally operator bounded by $K(P)$. The self-adjointness and core assertions then follow by Kato-Rellich, as for example shown in \cite{Moller}; the core used there is contained in $\mathfrak C \subseteq \cD(K(P))$. The same operator-bound gives equivalences of graph norms.
\end{proof}

 \begin{lemma}\label{lem:finite_cutoff_commutator}
For every $\Xi\in\mathfrak C$,
\begin{equation}\label{eq:finite_cutoff_commutator}
 a(k)H_\Lambda(P)\Xi
       =\bigl(H_\Lambda(P-k)+1\bigr)a(k)\Xi+g_\Lambda(k)\Xi
\end{equation}
as an identity in $L^2_{\mathrm{loc}}(\mathbb R^3_k;\mathcal F)$, and hence for almost every $k$. In particular, $a(k)\Xi\in\mathfrak C$ for every $k$ when the smooth component representatives are used, and $H_\Lambda(P)\Xi\in\mathcal D(\mathcal N^{1/2})$.
\end{lemma}

\begin{proof}
    The component formula for $a(k)$ and standard commutation relations give for $\xi \in \mathfrak C$
    $$a(k)K(P) \xi = (K(P-k) +1)a(k) \xi, \qquad a(k) \Phi(g_\Lambda) \xi = \Phi(g_\Lambda) a(k) \xi + g_\Lambda (k) \xi.$$
    Adding both equations up yields the formula \ref{eq:finite_cutoff_commutator}. Smoothness and compact support yield $a(k) \xi \in \mathfrak C$. Moreover, $H_\Lambda(P)\xi$ has only finitely many non-zero components, so the vector is contained in $\cD(\cN^{1/2})$. Finite particle number, compact support and $g_\Lambda \in L^2$ show that the claimed identity holds even in $L^2(\bbR^3;\cF)$.
\end{proof}

\begin{lemma}
\label{lem:finite_cutoff_commutator_gs}
    Let $0<\Lambda<\infty$ and let $\psi_{\Lambda,q}$ be a normalized
    ground state of $H_\Lambda(q)$. Suppose that $d_{\Lambda,q}=E_\Lambda(0)+1-E_\Lambda(q)>0$.
    Then
    \[
        a(k)\psi_{\Lambda,q}
        =
        -g_\Lambda(k)
        \bigl(H_\Lambda(q-k)+1-E_\Lambda(q)\bigr)^{-1}
        \psi_{\Lambda,q}
    \]
    as an identity in $L^2(\mathbb R^3;\mathcal F)$. In particular, $\|a(k)\psi_{\Lambda,q}\|
        \le \frac{|g_\Lambda(k)|}{d_{\Lambda,q}}$
    for almost every $k\in\mathbb R^3$.
\end{lemma}

\begin{proof}
    This is \cite[Proposition~B.5]{Moller} with $z=E_\Lambda(q)-1$.
\end{proof}

\end{document}